\documentclass[journal]{IEEEtran}
\ifCLASSINFOpdf
\else
\fi
\usepackage{amssymb,amsthm,pict2e,enumerate}
\usepackage{subfigure}
\usepackage{array}
\usepackage{booktabs}
\usepackage{color}
\usepackage{bm}
\usepackage{amsfonts}
\usepackage[cmex10]{amsmath}

\usepackage{cite}
 \usepackage{url}
 \usepackage{color}

\ifCLASSINFOpdf
   \usepackage[pdftex]{graphicx}
\else
   \usepackage[dvips]{graphicx}
\fi

\ifCLASSOPTIONcompsoc
\usepackage[caption=false,font=normalsize,labelfon
t=sf,textfont=sf]{subfig}
\else
\usepackage[caption=false,font=footnotesize]{subfi
g}
\fi

\newtheorem{Thm}{Theorem}
\newtheorem{Def}[Thm]{Definition}
\newtheorem{lemma}[Thm]{Lemma}
\newtheorem{corollary}[Thm]{Corollary}
\begin{document}

\title{Analysis of Frequency Agile Radar via Compressed Sensing}

\author{Tianyao Huang, 
Yimin Liu,~\IEEEmembership{Member,~IEEE,}
Xingyu Xu, 
Yonina C. Eldar,~\IEEEmembership{Fellow,~IEEE,}
Xiqin Wang
%        John~Doe,~\IEEEmembership{Fellow,~OSA,}
%        and~Jane~Doe,~\IEEEmembership{Life~Fellow,~IEEE}% <-this % stops a space
%\thanks{M. Shell was with the Department
%of Electrical and Computer Engineering, Georgia Institute of Technology, Atlanta,
\thanks{Partial results \cite{Huang2015} of this work were presented at the IEEE China Summit and International Conference on Signal and Information Processing, Chengdu, China, Sep. 2015. The work of T. Huang, Y. Liu, X. Xu and X. Wang was supported by the National Natural Science Foundation of China under Grant 61801258 and 61571260. (Corresponding author: Yimin Liu.)
\par
T. Huang, Y. Liu, X. Wang and X. Xu are with the Department of Electronic Engineering, Tsinghua University, Beijing 100084, China (e-mail:\{huangtianyao, yiminliu,  wangxq\_ee\}@tsinghua.edu.cn,  {\protect\url{xy-xu15@mails.tsinghua.edu.cn}}).
\par
Yonina C. Eldar is with the Department of Electrical Engineering, Technion-Israel Institute of Technology, Haifa 32000, Israel  (e-mail: {\protect\url{yonina@ee.technion.ac.il}}).
}
}

\maketitle

\begin{abstract}
%We analyze a frequency agile radar (FAR) via compressed sensing (CS) approaches. In the considered model, the FAR transmits pulses with random carrier frequencies, and aims at correct recovery of targets in the range-Doppler domain. Randomness in carrier frequencies results in a highly structured random sensing matrix. Properties of the sensing matrix are studied with tools developed in the CS community, and theoretical conditions are derived for the recoverability of targets. We prove bounds on the number of recoverable targets in the form of radar parameters. Numerical simulations and field experiments validate the theoretical findings and demonstrate the effectiveness of CS approaches.
Frequency agile radar (FAR) is known to have excellent electronic counter-countermeasures (ECCM) performance and the potential to realize spectrum sharing in dense electromagnetic environments. Many compressed sensing (CS) based algorithms have been developed for joint range and Doppler estimation in FAR. This paper considers  theoretical analysis of FAR via CS algorithms. In particular, we analyze the properties of the sensing matrix, which is a highly structured random matrix. We then derive bounds on the number of recoverable targets. Numerical simulations and field experiments validate the theoretical findings and demonstrate the effectiveness of CS approaches to FAR. 
\end{abstract}

%\begin{IEEEkeywords}
%Frequency agile radar, Random frequency radar, Compressed sensing, Performance guarantee, Spark.
%\end{IEEEkeywords}

% \fi
\IEEEpeerreviewmaketitle

\section{Introduction}
Frequency agile radars (FARs) are pulse-based radars, in which the carrier frequencies are varied in a random/pseudo-random manner from pulse to pulse as illustrated in Fig. \ref{fig:far}. Each transmission occupies a narrow band ($B_0$). Pulse returns of different frequencies are processed coherently to synthesize a wider band ($B>B_0$), which generates high range resolution (HRR) profiles.
\begin{figure}[!h]
\centering
\includegraphics[width=2.5in]{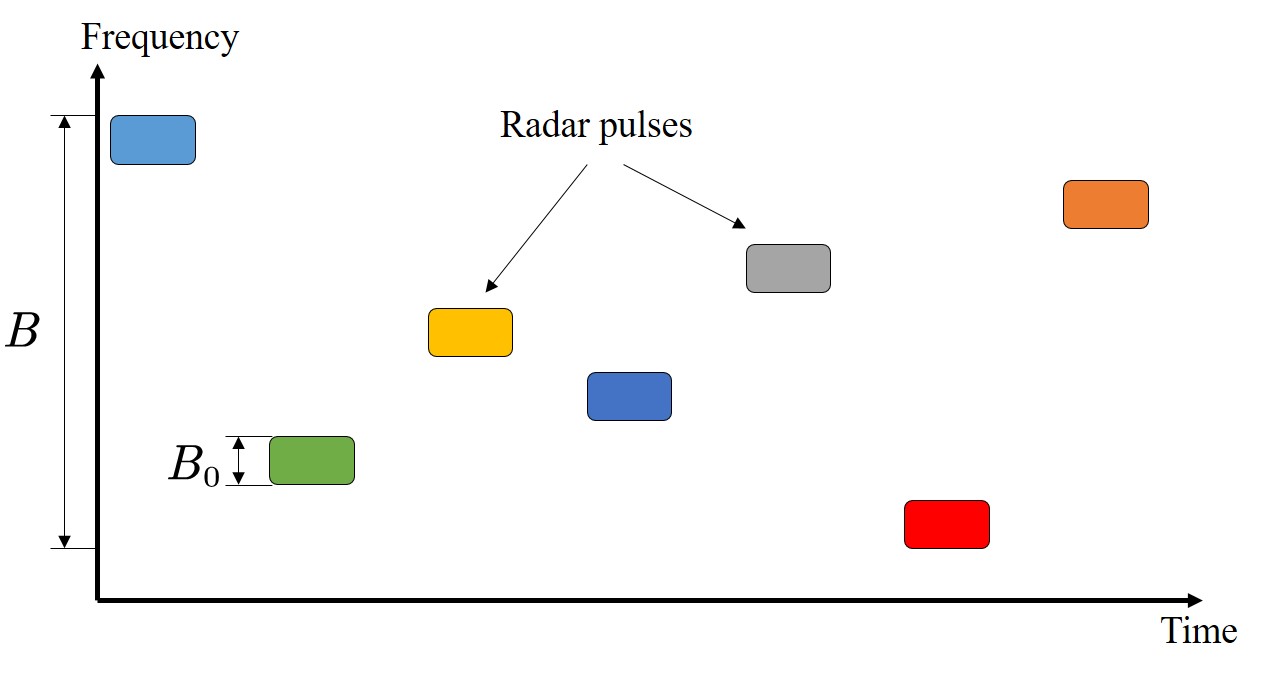}
\caption{An example of a FAR waveform. The boxes indicate the frequency band transmitted in the given time window.}
\label{fig:far}
\end{figure}
\par
Since the works \cite{liu2000sp, Axelsson2007}, frequency agility has received increasing attention \cite{Liu2008, Huang2012, Huang2014, Zhang2011, Yang2013,  Liu2014, Cohen2017a} in the radar community due to its multi-fold merits. 
First, frequency agility introduces good electronic counter-countermeasures (ECCM) performance, because the randomly varied frequencies of the pulses are difficult to track and predict. In addition, the flexibility of the narrow band transmission makes it easier to avoid and reject barrage jamming. 
Second, like stepped frequency radar, FAR can be used for two-dimensional (2D) imaging \cite{Yang2013, Liu2014},  while requiring only a narrow band receiver, which significantly lowers the hardware system cost. 
Third, in contrast to linearly stepped frequency radar, FAR decouples the range-Doppler parameters and produces a thumbtack ambiguity function \cite{Liu2008}. It also mitigates aliasing artifacts in synthetic aperture radar (SAR) \cite{Luminati2004} and extends the unambiguous Doppler window in inverse SAR (ISAR) imaging \cite{Huang2012}. 
Finally, frequency agility can be utilized to exploit vacant spectral bands \cite{Cohen2017a}, and shows potential to increase spectrum efficiency and cope with spectrum sharing issues in a contested, congested and competitive electromagnetic environment.
\par
We consider the problem of joint HRR profile and Doppler estimation of the target. %Note that ``range'' herein refers to high-range-resolution profile which owes to the wide, coherently synthesized bandwidth.
When a traditional matched filter is used for  joint range-Doppler estimation, sidelobe pedestal problems occur in FAR \cite{Liu2008}. Therefore, weak targets could be masked by the sidelobe of dominant ones, which restricts the application of FAR in target detection and feature extraction \cite{Liu2008}. 
By exploiting target sparsity, compressed sensing (CS) techniques \cite{Eldar2012,Eldar2015} have been applied in order to alleviate the sidelobe pedestal problem. 
{\it Liu et al.} \cite{Liu2008} propose the RV-IAP algorithm for joint range-Doppler estimation, which is based on the Orthogonal Matching Pursuit (OMP) method \cite{Tropp2007}. Since then, many practical CS algorithms for FAR have been developed \cite{Huang2012a,Huang2014}. 
\par
This paper focuses on the theoretical analysis of CS methods for FAR in terms of reconstruction performance. Theoretical conditions that guarantee perfect recovery in general CS have been extensively studied. Randomness plays a key role in many theoretical results, and often leads to good empirical results \cite{Eldar2012}. Near optimal conditions for Gaussian, sub-Gaussian, Bernoulli and random Fourier matrices have been derived \cite{Baraniuk2008a, Krahmer2014} (and references therein). However, the measurement matrix in FAR {differs from these random matrices, so that previous theoretical results are not directly applicable.}%In recent years, reconstruction guarantees have been developed for multiple-input–multiple-output (MIMO) radar system \cite{Rossi2014,Strohmer2014,Dorsch2017} and also for a random frequency diverse array (RFDA) radar system\cite{Liu2017}. These radar systems have signal models closer to that of an FAR, which partially inspires the theoretical analysis on the recoverability of the observed scene using FAR.
\par
We start by studying the measurement matrix properties of a FAR system. This matrix is random, due to the randomness of the frequencies. We begin by deriving probability bounds on the spark and coherence of the measurement matrix{, extending the results of \cite{Huang2015}}. Based on these bounds, we develop recovery guarantees for joint HRR profile and Doppler estimation using FAR systems. Theoretical results show that owing to the randomness of the carrier frequencies, with high probability, one can jointly obtain a HRR profile and Doppler of targets while transmitting narrow-band pulses. The number of recoverable targets is proved to be ${N}/{2}$ using $\ell_0$ minimization, or on the order of $\sqrt{\frac{N}{\log (N B/B_0)}}$ using $\ell_1$ minimization, where $N$ is the number of pulses. %{\color{red}The previous conference paper \cite{Huang2015} also presents some results on mutual incoherence and claims that the number of recoverable targets is at least on the order of $\sqrt{\frac{N}{\log (N B/B_0)}}$, however the spark property is not discussed and proofs are not provided.}
\par 
%Based on these properties, bounds on the number of recoverable targets, $K$, are derived, which dependent on the number of pulses, $N$, and the number of frequencies, $M$. According to CS theory \cite{Donoho2003},  to recover $K$ targets, a minimum requirement on the number of pulses is $N \geq 2K$. Here we show that for FAR $N = 2K$ pulses is almost surely sufficient to reconstruct $K$ targets, when the random carrier frequencies obey a continuous uniform distribution. When the random frequencies are discretely distributed, which is usually the case in practical scenarios, two sufficient conditions are presented for uniform and nonuniform recovery, respectively. 
We next perform simulations and field experiments to demonstrate the reconstruction performance of FAR using CS methods. We build an X-band FAR prototype with a synthetic bandwidth of 1 GHz, and {test} the recovery performance in a real environment. The results show that both the HRR profiles and Doppler of the observed target (a moving car) are reconstructed {with $N= 512$ and $B/B_0 = 32$}. 
\par
The rest of paper is organized as follows. Section \ref{sec:signalmodel} introduces the signal model and problem formulation. In Section \ref{sec:CS}, a brief review of CS algorithms and their performance guarantees is provided. We derive conditions for joint range-Doppler recovery using FAR in Section \ref{sec:FAR}. Numerical simulation and field experiment results are shown in Section \ref{sec:sim}. Section \ref{sec:con} concludes the paper. 
\par
Throughout the paper we use the following notation. The sets $\mathbb{C}$, $\mathbb{R}$, $\mathbb{Z}$, $\mathbb{N}$ refer to complex, real, integer, and natural numbers. 
Notation $| \cdot |$ is used for the modulus, absolute value or cardinality for a complex, real valued number, or a set, respectively, and $j:=\sqrt{-1}$. 
For $x \in \mathbb{R}$, $\lfloor x \rfloor$ {(or $\lceil x \rceil$)} is the largest {(smallest)} integer less {(greater)} than or equal to $x$. 
Uppercase boldface letters denote matrices (e.g., $\bm A$), and lowercase boldface letters denote vectors (e.g., $\bm a$).
The $m$,$n$-th element of matrix $\bm A$ is written as $[{\bm A}]_{m,n}$, and $[{\bm a}]_{n}$ denotes the $n$-th entry of a vector.  Given a matrix $\bm A \in \mathbb{C}^{M \times N}$, a number $n$ (or a set of integers, $\Lambda$), ${\bm A}_n$ (${\bm A}_{\Lambda} \in \mathbb{C}^{M \times |\Lambda|}$) denotes the $n$-th column of $\bm A$ (the sub-matrix consisting of the columns of $\bm A$ indexed by $\Lambda$). 
As for a vector $\bm a \in \mathbb{C}^{N}$,  ${\bm a}_{\Lambda} \in \mathbb{C}^{ |\Lambda|}$ denotes the sub-vector consisting of the elements of $\bm a$ indexed by $\Lambda$. 
The complex conjugate operator, transpose operator, and the complex conjugate-transpose operator are $^*$, $^T$, $^H$, respectively. 
We use $\| \cdot \|_p$, $p = 1, 2$ as the $\ell_p$ norm of an argument, and $\mathbb{P}(\cdot)$ denotes the probability of an event. 
Operations ${{\rm E}}[\cdot]$ and ${{\rm D}}[\cdot]$ represent the expectation and variance of a random argument, respectively. 
The real and imaginary part of a complex valued argument are denoted by ${\rm Re}\left( \cdot \right)$ and ${\rm Im}\left( \cdot \right)$, respectively.

\section{Signal Model}\label{sec:signalmodel}
\subsection{Radar Returns Model}
In this section, we introduce FAR, following the presentation in \cite{Huang2014}. A FAR system transmits monotone pulses, where the $n$-th transmitted pulse is written as
\begin{equation}
T_x(n,t) := {\rm rect}\left(\frac{t-nT_r}{T_p} \right) e^{j2\pi f_n (t-nT_r)},
\end{equation}
$n = 0,1,\dots,N-1$, where $T_r$ and $T_p$ are the pulse repetition interval and pulse duration, respectively, $T_r > T_p$, and rect$(\cdot)$ represents the rectangular envelope of the pulse
\begin{equation}
{\rm rect}(x):=\left\{ \begin{array}{l}
	1,\ 0\leq x\leq 1,\\
	0,\text{ otherwise}.\\
\end{array} \right. 
\end{equation}
The frequency of the $n$-th pulse $f_n$ is randomly varied as $f_n = f_c + d_nB$, where $f_c$ is the initial frequency, $d_n$ is the $n$-th random frequency-modulation code, $0 \leq d_n \leq 1$, and $B$ is the synthetic bandwidth. For a single pulse, the bandwidth {($B_0 = 1/T_p$) is narrow, $B_0<B$,} and the {coarse} range resolution {(CRR)} is $\frac{T_pc}{2}$,  where $c$ is the {speed of} light. Synthesizing echoes of different frequencies refines the range resolution to {$\frac{c}{2B}$}. We denote the number of HRR bins inside a CRR bin as 
\begin{equation}
{M :=\left\lceil \frac{T_pc}{2}\cdot \frac{2B}{c} \right\rceil = \lceil T_p B \rceil \in \mathbb{N}.}
\end{equation}
\par
Received echoes are {assumed} delays of the transmissions. We begin by assuming that there is a single ideal scatterer with scattering coefficient $\beta \in \mathbb{C}$. The echo of the $n$-th pulse can then be written as
\begin{equation}
R_x(n,t):  = \beta T_x \left(n,t-\frac{2r(t)}{c} \right),
\end{equation}
where $r(t)$ {denotes} the range of the scatterer with respect to the radar at time instant $t$. We assume that the scatterer is moving along the line of sight at a constant speed $v$, so that $r(t) = r(0) + vt$. After down conversion, the echo becomes
\begin{equation}\label{equ:sigmod_fasttime}
\begin{split}
R_d&(n,t):=R_x(n,t) \cdot e^{-j2\pi f_n (t-nT_r)} \\
&{=\beta {\rm rect}\left(\frac{t-{2r(t)}/{c}-nT_r}{T_p} \right) e^{j2\pi f_n (t-{2r(t)}/{c}-nT_r)}}\\
&\quad {\cdot e^{-j2\pi f_n (t-nT_r)}}\\
&{=\beta {\rm rect}\left(\frac{t-{2r(t)}/{c}-nT_r}{T_p} \right) e^{- j2\pi f_n \frac{2r(t)}{c} }}.
\end{split}
\end{equation}
\par
Echoes are sampled at the Nyquist rate of {a single pulse,} $f_s = 1/T_p${\color{red},} so that each echo pulse is sampled once. Every sample corresponds to a CRR bin{, and data from all CRR bins are processed in the same way}. %Some sub-Nyquist sampling methods are also applicable, but are out of the scope; refer to \cite{Bar-Ilan2014,Cohen2016, Xi2014}. 
Returns of {$N$} pulses from the same CRR bin are combined to a vector
\begin{equation}
\left[R_d(0,t), R_d(1,T_r+t),\dots,R_d\left(N-1,(N-1)T_r+t\right) \right],
\end{equation}
and processed to generate HRR profiles and {Doppler estimates}. {During the coherent processing interval (CPI), i.e. $NT_r$, we assume that the scatterer does not cross a CRR bin, which means that}
\begin{equation}
{vN T_r < \frac{T_pc}{2}.}
\end{equation}
\par
Without loss of generality, {suppose that} the $l$-th CRR bin contains the scatterer, $l =0,1, \dots, \left\lfloor{T_r}{f_s}\right\rfloor$. The corresponding sampling instant for the $n$-th pulse is $t = nT_r + l/f_s$.  Substituting $t = nT_r + l/f_s$ into (\ref{equ:sigmod_fasttime}), the sampled echoes are {given} by
\begin{equation}\label{equ:fast_time_samples}
\begin{split}
R_d&(n, nT_r + l/f_s)
=\beta e^{ -j4\pi (f_c + d_nB)  \left(r(0) + v(nT_r + l/f_s)\right)/{c} } \\
&\approx \beta e^{-j4\pi f_c \frac{r(0)+vl/f_s}{c}}e^{-j4\pi d_n B r(0)/c}e^{ -j4\pi f_cvT_r n \zeta_n/c },
\end{split}
\end{equation}
where the approximation holds if the term $e^{-j4\pi d_nBvl/(f_sc)} \approx 1$, which requires $e^{-j4\pi BvT_r/c} \approx 1$. {Here} $\zeta_n := 1+d_nB/f_c$. Generally, different carrier frequencies imply different Doppler shifts, unless the relative bandwidth $B/f_c$ is negligible, i.e. $\zeta_n \approx 1$. However, in a (synthetic) wideband radar, this approximation  does not usually hold, and could give rise to estimation performance deterioration in practice if applied. In the simulations and field experiments in Section \ref{sec:sim}, the signal processing algorithms do not adopt this assumption. However, in the mathematical analysis in Section \ref{sec:FAR}, we assume $\zeta_n \approx 1$ for theoretical convenience. The impact of the relative bandwidth will be discussed in the simulations. 
\par
For brevity, we omit the notation $l$, and write $R_d(n):= R_d(n, nT_r + l/f_s)$. We further introduce notations 
\begin{equation}
\left\{ \begin{array}{l}
	{\tilde \gamma} := \beta e^{-j4\pi f_c  \frac{r(0)+vl/f_s}{c}},\\
	{\tilde p} := -4\pi B r(0)/(Mc),\\
	{\tilde q} := -4\pi f_cvT_r  /c.\\
\end{array} \right. 
\end{equation}
With these definitions (\ref{equ:fast_time_samples}) becomes
\begin{equation}\label{equ:fast_time_samples_pq}
\begin{split}
R_d(n) \approx {\tilde \gamma}e^{j{\tilde p}Md_n+j{\tilde q} n \zeta_n }.
\end{split}
\end{equation}
After the unknowns ${\tilde \gamma}$, ${\tilde p}$ and ${\tilde q}$ are estimated, the absolute intensity $|\beta|$, HRR range $r(0)$ and velocity $v$ are inferred as $|{\tilde \gamma}|$, $-\frac{Mc{\tilde p}}{4\pi B}$ and $-\frac{c{\tilde q}}{4\pi f_c T_r}$, respectively.  
\par
When there are $K$ scatterers occurring inside the CRR cell, radar returns are modeled as a combination of returns from all scatterers,
\begin{equation}\label{equ:gamma_p_q}
R_d(n) = \sum_{k=0}^{K-1} {\tilde \gamma}_k e^{j{\tilde p}_kMd_n + j{\tilde q}_kn\zeta_n},
\end{equation}
where ${\tilde \gamma}$, ${\tilde p}$ and ${\tilde q}$ in (\ref{equ:fast_time_samples_pq}) are replaced with ${\tilde \gamma}_k$, ${\tilde p}_k$ and ${\tilde q}_k$ for the $k$-th scatterer, respectively.
\par
{To} avoid grating lobes in the HRR profiles (which are also called ghost images in the literature) \cite{Liu2009, Liu2014a}, the frequency codes are required to satisfy $\min_{n \neq m} |d_n - d_m| \leq 1/M$, $n,m = 0,1,\dots,N-1$. When codes are discrete, {we denote by $\mathcal{D}_d$ the set of available frequency codes and by $M^{\star}:=\left|\mathcal{D}_d\right|$ the number of codes. The codes are often uniformly spaced, e.g. $d_n \in \mathcal{D}_d :=\left\{ \frac{m}{M^{\star}} | m = 0,1,\dots, M^{\star}-1 \right\}$. It is required that $M^{\star} \geq M$ and} a typical choice is {$M^{\star} = M$}. When codes belong to  a continuous set $\mathcal{D}_c := [0,1)$ (for some of our theoretical results), the requirement is usually easy to satisfy with $N\geq M$. We assume that in both discrete and continuous cases, the codes $d_0, \dots d_{N-1}$ are identically, independently, and uniformly distributed.
\subsection{Signal Model in Matrix Form}\label{subsec:sm_matrix}
We {can} rewrite (\ref{equ:gamma_p_q}) in matrix form as
\begin{equation}{\label{equ:noiselessSigModel}}
{\bm y} ={\bm \Phi}{\bm x},
\end{equation}
where the measurement vector $\bm y \in \mathbb{C}^N$ has entries $[\bm y]_n = R_d(n)$. The vector $\bm x  \in \mathbb{C}^{NM}$ corresponds to the scattering intensities $\tilde \gamma$. The pair $({\tilde p},{\tilde q})$ defines the key target parameters, range and Doppler, and belongs to a continuous 2D domain. The resolutions for ${\tilde p}$ and ${\tilde q}$ are {$\frac{2\pi}{M}$} and {$\frac{2\pi}{N}$}, respectively. Consider the unambiguous continuous region $(p,q) \in [0,2\pi)^2$, and discretize $p$ and $q$ at the Nyquist rates, {$\frac{2\pi}{M}$} and {$\frac{2\pi}{N}$}, respectively. Thus, one obtains $p_m := \frac{2\pi m}{M}$ and $q_n := \frac{2\pi n}{N}$, $m = 0,1,\dots, M-1$, $n = 0,1,\dots, N-1$. Denote the sets containing HRR grids and Doppler grids as $\mathcal{P} := \left\{  \left.\frac{2\pi m}{M}  \right| m = 0,1,\dots, M-1 \right\}$ and $\mathcal{Q} := \left\{  \left.\frac{2\pi n}{N}  \right| n = 0,1,\dots, N-1 \right\}$, respectively, and assume that the targets are located precisely on the grid. 
Define {the matrix} ${\bm X} \in \mathbb{C}^{M \times N}$ with entries
\begin{equation}
[\bm X]_{m,n}=\left\{ \begin{array}{l}
	{\tilde \gamma}_k, \text{ if }\exists k, \text{ } \left({\tilde p}_k, {\tilde q}_k \right) = \left( p_m, q_n\right),\\
	0,\text{ otherwise},\\
\end{array} \right. 
\end{equation}
representing the 2D scattering coefficients in the range-Doppler domain, $m = 0,1,\dots,M-1$ and $n = 0,1,\dots,N-1$. We vectorize $\bm X$ to obtain ${\bm x}:={\rm vec}(\bm X^T)$ with entries $[\bm x]_{n+mN} := [\bm X]_{m,n}$.
\par
{To introduce} the measurement matrix $\bm \Phi \in \mathbb{C}^{N \times MN}$, we define the matrices ${\bm R} \in \mathbb{C}^{N \times M}$ and ${\bm D}  \in \mathbb{C}^{N \times N}$, corresponding to HRR range and Doppler parameters, respectively, with entries
\begin{equation}\label{equ:R_matrix}
[{\bm R}]_{n,m} :=  e^{jp_m Md_n},
\end{equation}
\begin{equation}
[{\bm D}]_{n,l} :=  e^{jq_ln\zeta_n},
\end{equation}
$m = 0,1,\dots,M-1$, and $l,n = 0,1,\dots,N-1$. If $\zeta_n \approx 1$, then ${\bm D}$ is a Fourier matrix. Define ${\bm \Phi} := \left({\bm R}^T \circledcirc {\bm D}^T\right)^T $, where $\circledcirc$ denotes the Khatri-Rao product. Then the elements of ${\bm \Phi}$ are given by
\begin{equation}\label{equ:phi}
\left[{\bm \Phi} \right]_{n,l+mN} := \left[{\bm R} \right]_{n,m}\left[{\bm D} \right]_{n,l}
= e^{jp_mMd_n + jq_l n\zeta_n},
\end{equation}
$m = 0,1,\dots,M-1$ and $l,n = 0,1,\dots,N-1$. When echoes are corrupted by additive noise ${\bm w} \in \mathbb{C}^N$, (\ref{equ:noiselessSigModel})  becomes
\begin{equation}{\label{equ:noisySigModel}}
{\bm y} ={\bm \Phi}{\bm x} +{\bm w}.
\end{equation}
\par
The sensing matrix ${\bm \Phi}$ in (\ref{equ:noisySigModel}) has more columns than rows, $MN \geq N$, which shows that joint range and Doppler estimation in FAR is naturally an under-determined problem. When $\bm x$ is $K$-sparse, which means there are $K$ non-zeroes in $\bm x$, and $K \ll MN$, CS algorithms can be used to {solve} (\ref{equ:noisySigModel}). The targets' parameters can then be recovered from the support set of ${\bm x}$. 
\subsection{Discussion on the Signal Model}
\par
Note that when there is only one scatterer observed, the matched filter that maximizes the signal to noise ratio (SNR) works well in FAR. However, when there are multiple scatterers, sidelobe pedestal problem occurs and weak targets can be masked by the dominant targets' sidelobe. The matched filter estimates the scattering intensities by 
\begin{equation}\label{equ:mf}
\hat{\bm x} := \bm \Phi^H \bm y = \bm \Phi^H {\bm \Phi}{\bm x} + {\bm \Phi^H}{\bm w}.
\end{equation}
In such an under-determined model, $\bm \Phi^H {\bm \Phi} \neq \bm I$, spurious responses emerge in $\hat{\bm x}$ even if there is no noise, i.e. $\bm w = \bm 0$. These spurious responses are the sidelobe pedestal. 
\par
To better interpret the sidelobe pedestal problem, we compare the signal model of FAR with that of an instantaneous wideband radar (IWR). In such a hypothetical radar, we assume that the radar transmits/receives all of its $M$ sub-bands {(with $\mathcal{D}_d$ as the set of frequency codes, $\left| \mathcal{D}_d\right|=M^{\star}=M$)}, and processes the echoes individually for each band. {In FAR, the same set $\mathcal{D}_d$ is also applied with $Md_n \in \mathbb{N}$.} In analogy to (\ref{equ:gamma_p_q}), the return of the $m$-th frequency in the $n$-th pulse can be written as
\begin{equation}\label{equ:IWR}
R_{\rm IWR}(m,n) := \sum_{k=0}^{K-1} {\tilde \gamma}_k e^{j{\tilde p}_k m + j{\tilde q}_kn\eta_{m}},
\end{equation}
where {$\eta_m := 1+\frac{mB}{Mf_c}$}, $m = 0,1,\dots,M-1$, and $n = 0,1,\dots,N-1$. For notational brevity and simplicity, we assume $\eta_{m} \approx 1$ and $\zeta_{n} \approx 1$ for IWR and FAR, respectively. In this case, (\ref{equ:IWR}) can be rewritten in matrix form as
\begin{equation}\label{equ:sm_matrix}
{\bm Z} = \bm F \bm X \bm D^T,
\end{equation}
where ${\bm Z} \in \mathbb{C}^{M \times N}$ {has} entries $[{\bm Z}]_{m,n} = R_{\rm IWR}(m,n)$, and ${\bm F} \in \mathbb{C}^{M \times M}$ is a Fourier matrix with entries $[{\bm F}]_{l,m} :=  e^{jp_m l}$, $m,l = 0,1,\dots,M-1$, $n = 0,1,\dots,N-1$. Equivalently, 
\begin{equation}\label{equ:sm_vector}
{\bm z} = \left( \bm F \otimes \bm D \right) \bm x,
\end{equation}
where {$\bm z := {\rm vec}(\bm Z^T) \in \mathbb{C}^{MN}$} and $ \otimes$ denotes the Kronecker product. The sensing matrix in the IWR {$\bm \Psi := \bm F \otimes \bm D \in \mathbb{C}^{MN \times MN}$} is orthogonal%\footnote{Note that when the assumption $\zeta_n \approx 1$ does not hold, the sensing matrix of the IWR cannot be written in the form of $\bm R \otimes \bm D$, and is generally not orthogonal.}
, i.e. {$\frac{1}{MN}\bm \Psi^H \bm \Psi = \bm I$}, and the sidelobe pedestal problem vanishes.
\par
The measurements in FAR can be regarded as sampling\footnote{{{Since an instantaneous narrowband waveform is used in FAR, it naturally enjoys low data rate in comparison with IWR. However, this paper does not aim at minimizing the data rate. It may have the potential to further reduce the data rate by combining frequency agility with approaches like sub-Nyquist sampling in the fast-time domain \cite{Bar-Ilan2014, Cohen2017a, Xi2014}, omitting some pulses or frequency bands \cite{Hu2011,Cohen2017}.}}} of the IWR measurements, i.e.
\begin{equation}
{[\bm y]_n = [\bm Z]_{Md_n,n}}, \ n = 0,1,\dots, N-1.
\end{equation}
Only one sub-band data is acquired for each pulse. Therefore the sensing matrix of FAR consists of partial rows of that of {the} IWR, i.e.,
\begin{equation}
{\left(\bm \Phi^T \right)_n = \left(\bm \Psi^T \right)_{n+Md_nN}}, \ n = 0,1,\dots, N-1,
\end{equation}
and becomes an under-determined matrix. This interpretation suggests that the sidelobe pedestal of FAR results from the information loss in the frequency domain. The spectral incompleteness leads to an under-determined problem (\ref{equ:noiselessSigModel}). 
In Section \ref{sec:FAR}, we prove that, owing to the randomness of the frequencies, the scatterers can still be correctly reconstructed {via CS methods} with high probability. 
\par  

%-------------------------------------------------------------------------------------------------------------------------------
\section{Review of Compressed Sensing}\label{sec:CS}
In Section \ref{sec:FAR}, we prove that using CS methods, FAR can {provably} recover  the HRR profiles and Doppler. Before deriving the results, we review some basic {notions} of CS \cite{Eldar2015}.
\par
Consider an under-determined linear regression problem, e.g. (\ref{equ:noiselessSigModel}), where $\bm x$ is sparse. The sparsest solution can be obtained via
\begin{equation}\label{equ:P0}
\min_{\bm x} \| \bm x\|_0,{\text{ }}s.t. {\rm \ }  {\bm y} = {\bm A}{\bm x}, \tag{$P_0$}
\end{equation}
where $\| \cdot \|_0$ denotes $\ell_0$ ``norm'' of a vector, i.e. the number of non-zeroes. This solution is the true vector, when the sensing matrix $\bm A$ has the {spark} property. 
\begin{Def}[Spark, \cite{Eldar2015}]
Given a matrix $\bm A$, Spark($\bm A$) is the smallest possible number such that there exists a subgroup of columns from $\bm A$ that are linearly dependent. 
\end{Def}
Unique recovery of $\bm x$ can be ensured if the following condition is satisfied.
\begin{Thm}\label{thm:spark}
The equation ${\bm y} = {\bm A}{\bm x}$ {is} uniquely solved by (\ref{equ:P0}) if and only if $\| \bm x\|_0 < \frac{ {\rm Spark}(\bm A)}{2}$.
\end{Thm}
The above theorem provides a fundamental limit on the maximum sparsity that leads to successful recovery. In general, $\ell_0$ optimization is NP-hard. A widely used alternative is basis pursuit, which solves the problem
\begin{equation}\label{equ:P1}
\min_{\bm x} \| \bm x\|_1, s.t. {\rm \ }  {\bm y} = {\bm A}{\bm x}. \tag{$P_1$}
\end{equation}   
In noisy cases, variants like basis pursuit denoising, LASSO and Dantzig selector can be applied. Many greedy methods have also been suggested {to approximate (\ref{equ:P0})}.% \cite{Foucart2013}. 
%There are also some greedy approaches; refer to \cite{Foucart2013} and references therein.
\par
Sufficient conditions that guarantee uniqueness using these methods are extensively studied. Bounds on the mutual incoherence property (MIP) and restricted isometry property (RIP) are widely applied conditions to ensure sparse recovery.
In this paper, we rely on the MIP. A matrix $\bm A$ has MIP if its coherence is small, where coherence is defined as the maximum correlation between two columns, i.e.
\begin{equation}
\mu({\bm A}) := \max_{l \neq k}  \frac{\left|{\bm A}_l^H {\bm A}_k\right|}{\|{\bm A}_l\|_2\|{\bm A}_k\|_2}. 
\end{equation}
\par
\begin{Thm}[\hspace{1sp}\cite{Fuchs2004}]{\label{thm:mip}}
If a matrix ${\bm A} \in \mathbb{C}^{N \times L}$ has coherence $\mu(\bm A)<\frac{1}{2K-1}$, then for any ${\bm x} \in \mathbb{C}^L$ of sparsity $K$, ${\bm x}$ is the unique solution to (\ref{equ:P1}).
\end{Thm}
The condition in Theorem \ref{thm:mip} ensures recovery in the presence of noise and also recovery using a variety of {computationally efficient} methods{\cite{Eldar2012}}.
\par
\section{Sensing Matrix Properties of FAR}\label{sec:FAR}
In this section, we analyze the {spark} and MIP properties of {the }FAR's sensing matrix. These results are then used together with Theorems \ref{thm:spark} and \ref{thm:mip} to establish performance guarantees for FAR. In the following derivations, we assume that $\zeta_n \approx 1$.
%({\color{red} say why these results are good.})
\subsection{Spark Property}
%In this subsection, we derive the performance guarantee of FAR using the Spark property. 
The following theorem proves that the sensing matrix of FAR almost surely has the {spark} property. 
\begin{Thm}\label{thm:sparkFAR}
Consider $\bm \Phi \in \mathbb{C}^{N \times NM}$ defined in (\ref{equ:phi}) with $d_n$ drawn independently from a uniform continuous distribution over $\mathcal{D}_c = [0,1)$, $n = 0,1,\dots,N-1$. Then, with probability $1$, Spark$(\bm \Phi) = N + 1$.
\end{Thm}
\begin{proof}
See Appendix \ref{app:spark}.
\end{proof}
Since $\bm \Phi$ has $N$ rows, there must be a linearly dependent submatrix with $N + 1$ columns. Owing to the randomness of the carrier frequencies, Theorem \ref{thm:sparkFAR} shows that a sub-matrix built from any $N$ columns of $\bm \Phi$ is of full rank almost surely. The result is based on the assumption of a continuous distribution on $d_n$. An immediate consequence of Theorems \ref{thm:spark} and \ref{thm:sparkFAR} is the following corollary.
\begin{corollary}\label{thm:limit}
Consider a FAR whose frequency modulation codes are drawn independently from a uniform continuous distribution over $\mathcal{D}_c = [0,1)$, $n = 0,1,\dots,N-1$. Then, with probability 1, $K_{\max}= \frac{N}{2}$ scatterers can be exactly recovered by  (\ref{equ:P0}), where $N$ is the number of pulses.
\end{corollary}
\subsection{Mutual Incoherence Property}\label{subsec:MIP}
To obtain performance guarantees using $\ell_1$ minimization or greedy CS methods under noiseless/noisy environments, we derive the MIP for FAR. We start by analyzing the asymptotic statistics of the FAR's sensing matrix. Then invoking Theorem \ref{thm:mip}, we obtain the maximum number of scatterers that FAR {guarantees} to exactly reconstruct with high probability.
\par
Assume in this subsection that $d_n \sim U\left( \mathcal{D}_d\right)$, {$\left| \mathcal{D}_d\right|=M^{\star}=M$}, and recall that the parameters $p$ and $q$ are on a grid, i.e. $p\in \mathcal{P}$ and $q \in \mathcal{Q}$. First, consider the Gram matrix $\bm \Phi^H \bm \Phi$, which links to the coherence. Define  $\bm G \in \mathbb{R}^{NM \times NM}$ as the modulus matrix of the Gram matrix, i.e.
\begin{equation}
[\bm G]_{k,l} = \left|\left[\bm \Phi^H \bm \Phi \right]_{k,l}\right|,k,l=0,1,\dots,NM-1.
\end{equation}
We then have the following results, some of which are partially inspired by \cite{Rossi2014}.
\begin{lemma}{\label{lem:toep}}
The rows of the modulus matrix $\bm G$ are permutations of elements in its first row.
\end{lemma}
\begin{proof}
Denote by $\bm \Phi_{l_1}$ and $\bm \Phi_{l_2}$, $l_1,l_2 = 0, 1,\dots,MN-1$, two columns in $\bm \Phi$, corresponding to $(p_{m_1},q_{k_1})$ and $(p_{m_2},q_{k_2})$, respectively, $m_1,m_2 = 0,1,\dots,M-1$, $k_1,k_2 = 0,1,\dots,N-1$. Then
\begin{equation}\label{equ:grammatrix}
\begin{split}
{\bm \Phi}_{l_1}^{H} {\bm \Phi}_{l_2} &= \sum \limits_{n = 0}^{N-1} e^{-jp_{m_1}Md_n - jq_{k_1}n} \cdot e^{jp_{m_2}Md_n + jq_{k_2}n} \\
&= \sum \limits_{n = 0}^{N-1} e^{-j(p_{m_1}-p_{m_2})Md_n - j(q_{k_1}-q_{k_2})n} \\
&= \sum \limits_{n = 0}^{N-1} e^{-j2\pi({m_1 - m_2})d_n - j\frac{2\pi({k_1 - k_2})}{N}n}.
\end{split}
\end{equation}
Clearly (\ref{equ:grammatrix}) depends only on the difference between grid points, i.e. $m_1-m_2 \in \{-M+1,\dots,M-1\}$ and $ k_1-k_2 \in \{ -N+1,\dots,N-1 \}$, and is independent of the particular indices $l_1$ and $l_2$. In addition, $\left|{\bm \Phi}_{l_1}^{H} {\bm \Phi}_{l_2} \right| = \left| {\bm \Phi}_{l_2}^{H} {\bm \Phi}_{l_1}\right|$. Therefore, for any element in $\bm G$, one can find an element in the first row of $\bm G$ with the same value.
\end{proof}
Consider now the $l$-th element in the $0$-th row of $\bm G$, $l \neq 0$, which corresponds to the $l$-th column of $\bm \Phi$. Note that each column of $\bm \Phi$ relies on a specific parameter pair $(p,q)$. For notational simplicity, we drop the subscripts of $(p,q)$ related to $\bm \Phi_l$, and define 
\begin{equation}\label{equ:chi0}
\chi_l := \frac{1}{N}\bm \Phi_0^{ H}\bm \Phi_l =  \frac{1}{N} \sum \limits_{n = 0}^{N-1} e^{jpMd_n + jqn},~l=1,\dots,NM-1.
\end{equation}
We now analyze the mutual coherence, $\mu = \max_{l \neq 0} \left| \chi_l \right|$. Since $d_n$ is random, $\chi_l$ is also random unless $p =0$, in which case  $\chi_l$ reduces to a constant $\frac{1}{N} \sum \limits_{n = 0}^{N-1} e^{ jqn} = 0$ for $q \in \mathcal{Q} \backslash \{0\}$. This constant does do not affect the value of $\mu$ and is thus ignored. Define a set excluding these constants as
\begin{equation}
\begin{split}
\Xi :&=\{1, 2, \dots, NM-1\} \backslash \left\{1, \dots, N-1\right\} \\
&= \{N, N+1, \dots, NM-1\}.
\end{split}
\end{equation}
Then $\chi_l$ is a random variable, $ l \in \Xi$, and has the following statistical characteristics. 
\begin{lemma}
\label{lem:asymGaussian}
As $N \rightarrow \infty$, the real and imaginary parts of $\chi_l$, ${\rm Re}\left( \chi_l \right)$ and ${\rm Im}\left( \chi_l \right)$, $l \in \Xi$, have a joint Gaussian distribution,
\begin{equation}
\label{equ:cs:zeromean}
\left[
\begin{array}{c}
{\rm Re}\left( \chi_l \right) \\
{\rm Im}\left( \chi_l \right) \\
\end{array}
 \right] \sim \mathcal{N}
\left( \left[
\begin{array}{c}
0\\
0\\
\end{array}
\right],\left[
\begin{array}{cc}
\frac{1}{2N} & 0\\
0 & \frac{1}{2N}\\
\end{array}
\right]
\right),
\end{equation}
except in the special case that the corresponding parameters $p = q = \pi${. In this setting}, the joint Gaussian distribution becomes
\begin{equation}
\label{equ:cs:zeromean2}
\left[
\begin{array}{c}
{\rm Re}\left( \chi_l \right) \\
{\rm Im}\left( \chi_l \right) \\
\end{array}
 \right] \sim \mathcal{N}
\left( \left[
\begin{array}{c}
0\\
0\\
\end{array}
\right],\left[
\begin{array}{cc}
{\frac{1}{N}} & 0\\
0 & 0\\
\end{array}
\right]
\right).
\end{equation}
\end{lemma}
\begin{proof}
See Appendix \ref{app:asymGaussian}.
\end{proof}
The special case $p = q = \pi$ corresponds to a specific grid point of the $(p,q)$ plane{. Considering the generic case and this special case separately leads to the following conclusions.}
\begin{corollary}\label{cor:Rayleigh}
When $N \rightarrow \infty$ and $l \in \Xi$, {with $N\epsilon^2>2/\pi$,} 
\begin{equation}\label{equ:rayleigh}
\mathbb{P}\left( \left| \chi_l \right| > \epsilon \right) \leq e^{-N\epsilon^2{/2}}.
\end{equation}
\end{corollary}
\begin{proof}
Lemma \ref{lem:asymGaussian} proves that when $p$ and $q$ do not equal $\pi$ simultaneously, the real and imaginary parts of $\chi_l$ asymptotically obey $\mathcal{N}(0,\frac{1}{2N})$ independently. Therefore, the magnitude $\left|\chi_l\right|$ obeys a Rayleigh distribution with probability density function $f(x) = 2Nxe^{-Nx^2}$, $x \geq 0$, and cumulative distribution function $F(x)=1-e^{-Nx^2}$, $x \geq 0$. Thus, $\mathbb{P}\left( \left| \chi_l \right| > \epsilon \right)=1-F(\epsilon)$, which yields 
\begin{equation}
\mathbb{P}\left( \left| \chi_l \right| > \epsilon \right) = e^{-N\epsilon^2}{ \le e^{-N\epsilon^2/2}}.
\end{equation}
In the special case that $p=q=\pi$,  the real part of $\chi_l$ asymptotically obeys $\mathcal{N}(0,{\frac{1}{N}})$ independently and the imaginary part vanishes. {Then elementary estimates of the Gaussian error function yield
\begin{equation}
\mathbb{P}(|\chi_l| > \epsilon) \le \sqrt{\frac{2}{\pi N\epsilon^2}}e^{-N\epsilon^2/2},
\end{equation}
which is less than $e^{-N\epsilon^2/2}$ if $\frac{2}{\pi N\epsilon^2}\le1$, i.e. $N\epsilon^2>\frac{2}{\pi}$}. 
\end{proof}
\begin{lemma}\label{cor:unionbound}
The maximum $\mu = \max_l | \chi_l|$, $l \in \Xi$, satisfies the following as $N \rightarrow \infty$,
\begin{equation}\label{equ:pr_union}
\mathbb{P}\left( \mu > \epsilon \right) \leq  (MN-N) e^{-N\epsilon^2{/2}}.
\end{equation}
\end{lemma}
\begin{proof}
{For fixed $\epsilon>0$, we have $N\epsilon^2>\frac{2}{\pi}$ as $N\to\infty$.} According to the union bound
\begin{equation}\label{equ:ccdf}
\begin{split}
\mathbb{P} (\mu > \epsilon) &\leq \sum \limits_{l \in \Xi} \mathbb{P}(|\chi_l|>\epsilon)\\
&\leq (MN-N) e^{-N\epsilon^2{/2}},
\end{split}
\end{equation}
since there are $NM - N$ indices in $\Xi$.
\end{proof}
\par
%According to Theorem \ref{thm:mip}, when $\mu ({\bf \Phi}) < \frac{1}{2K-1}$, $\ell_1$ minimization promises exact recovery. 
We next derive a condition for FAR to meet the requirement of Theorem \ref{thm:mip} $\mu ({\bf \Phi}) < \frac{1}{2K-1}$ with high probability.
\begin{Thm}\label{thm:MIP}
The coherence of $\bf \Phi$, defined in (\ref{equ:phi}), obeys $\mu ({\bf \Phi}) < \frac{1}{2K-1}$ with a probability higher than $1-\delta$, when
\begin{equation}\label{equ:K_rangeMIP}
K \leq \frac{1}{2{\sqrt 2}}\sqrt{\frac{N}{\log (MN-N) - \log \delta}}+ \frac{1}{2}.
\end{equation}

\end{Thm}
\begin{proof}
Let {$\epsilon=\frac{1}{2K-1}$. From (\ref{equ:K_rangeMIP}), we have that
\begin{equation}\label{equ:tmp_Ne_bound}
  N\epsilon^2\ge 2 \left(\log (MN-N) - \log \delta \right).
\end{equation}
Assume that $0<\delta<1$ and $MN-N=N(M-1)\ge 2$. Then
\begin{equation}
  N\epsilon^2\ge 2\log 2>\frac{2}{\pi}.
\end{equation}
Using (\ref{equ:pr_union}) and (\ref{equ:tmp_Ne_bound}), we finally obtain
\begin{equation}
\mathbb{P}\left( \mu \leq \epsilon \right) > 1-(MN-N)e^{-N\epsilon^2/2}\ge 1-\delta,
\end{equation}
completing the proof.
}
\end{proof}
%{\color{red} ``compare to other radar results in the literature''.}
Theorem \ref{thm:MIP} {shows} that the sensing matrix of FAR has the MIP; thus, according to Theorem \ref{thm:mip}, HRR range-Doppler reconstruction is guaranteed if the number of targets satisfies $K = O\left( \sqrt{\frac{N}{\log MN}}\right)$, where $N$ in the numerator represents the number of measurements, and $MN$ in the denominator links to the number of grid points. The work \cite{Rossi2014} for direction of arrival estimation using random array MIMO radar also proposes a bound based on MIP, which {guarantees} recovery of a number of targets on the order of $K = O\left( \frac{\sqrt{L}}{\log G}\right)$, where $L$ and $G$ are the number of measurements and grid points, respectively. In \cite{Rossi2014}, a bound for non-uniform recovery based on RIPless theory is also provided, and $K$ is relaxed to $K = O\left( \frac{L}{\log^2 G}\right)$. However, RIPless is not directly applicable to FAR, because for each row of $\bm \Phi$, i.e. $\bm a = \left( \bm \Phi^T \right)_n \in \mathbb{C}^{MN}$, ${\rm E}\left[ \bm a \bm a^H \right]$ is rank deficient and the isotropy property ${\rm E}\left[ \bm a \bm a^H \right] = \bm I$ does not hold. Dorsch and Rauhut \cite{Dorsch2017} analyze the joint angle-delay-Doppler recovery performance using MIMO based on RIP. They assume a periodic random probing signal with $N_t$ independent samples. {In this case} the recoverable number of scatterers is on the order of $K = O\left( \frac{N_t}{\log^2 G}\right)$. {Though the} RIP leads to a tighter bound than MIP, the RIP {of the} FAR {system matrix} is still an open question.
\par
Sufficient conditions that guarantee uniform recovery are usually pessimistic. It is well known that CS algorithms often outperform the theoretical uniform recovery guarantees. In the next section, we evaluate the practical performance of FAR using CS methods.

\section{Simulation and Experimental Results}\label{sec:sim}
In this section, simulations and field experiments are executed to demonstrate the properties of the sensing matrix of FAR and the effectiveness of CS algorithms to reconstruct the targets' HRR range and Doppler.
\subsection{Spark Property}
First, the {spark} property of the sensing matrix ${\bm \Phi} \in \mathbb{C}^{N \times NM}$ is discussed. We construct a sub-matrix ${\bm \Phi}_{\Omega} \in \mathbb{C}^{N \times N}$ of $\bm \Phi$, where the set $\Omega \subset \{0,1,\dots,NM-1\}$ and $|\Omega| = N$, and calculate the minimum singular value of ${\bm \Phi}_{\Omega}$. We check whether it is equal to zero (rank deficient). Concretely, we set $N =6$ and $M =3$, which are small to make it possible to enumerate all the $\binom{NM}{N}$ sub-matrices of ${\bm \Phi}$. We record the minimum singular value $\sigma_N$ (normalized by $\sqrt{N}$) of each sub-matrix; thus, we obtain $\binom{NM}{N}$ results, among which the minimum is denoted as $\sigma_{\Omega}$.
The frequency codes $d_n$ are distributed uniformly on a continuous set. We further assume the relative bandwidth satisfies $B/f_c \approx 0$. We perform 2000 Monte-Carlo trials. The histograms of $\sigma_N$ and $\sigma_{\Omega}$ are depicted in Fig. \ref{fig:hist_sigma_N} and Fig. \ref{fig:hist_sigma_omega}, respectively. The minimum of $\sigma_N$ and $\sigma_{\Omega}$ is $1.28 \times 10^{-6}>0$. The results indicate that a continuous distribution of codes results in good properties of the sensing matrix. For comparison, we also perform simulations with codes distributed on the discrete set $\mathcal{D}_d$, and count the number of minimum singular values $\sigma_N$, i.e. $\sigma_{\Omega}$, less than $\epsilon_{\rm SVD} = 1 \times 10^{-15}$, which leads to Pr$(\sigma_{\Omega}<\epsilon_{\rm SVD}) \approx 0.358$. Therefore, a continuous distribution leads to better spark performance than a discrete distribution.
\begin{figure}[!ht]
\centering
\includegraphics[width=2.5in]{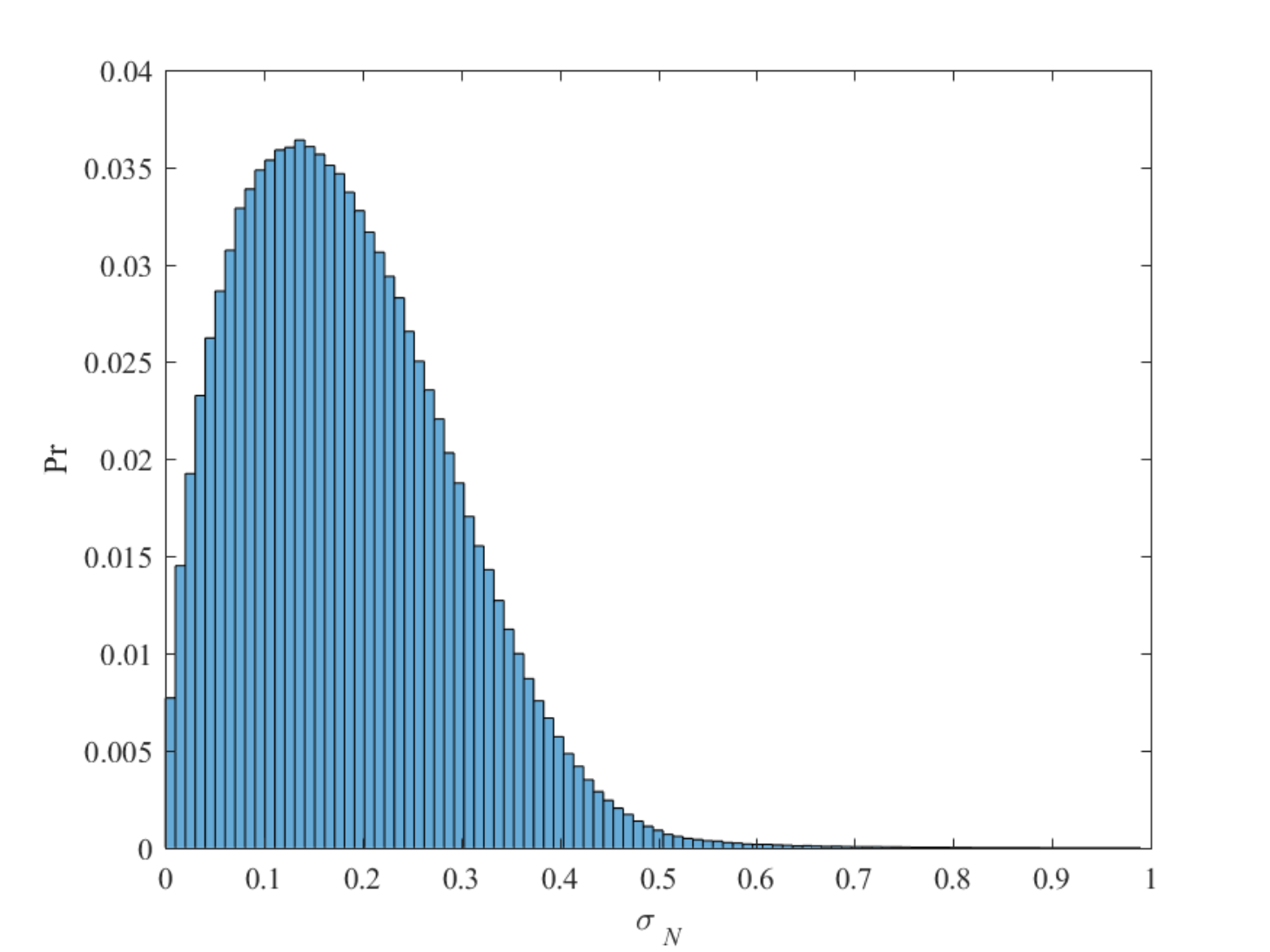}
\caption{Histogram of $\sigma_{N}$ (normalized by $\sqrt{N}$), the minimum singular value of each $N \times N$ sub-matrix. The histogram is obtained using 2000 Monte-Carlo trials and $\binom{NM}{N}$ sub-matrices in each Monte-Carlo trial.}
\label{fig:hist_sigma_N}
\end{figure}
\par
\begin{figure}[!ht]
\centering
\includegraphics[width=2.5in]{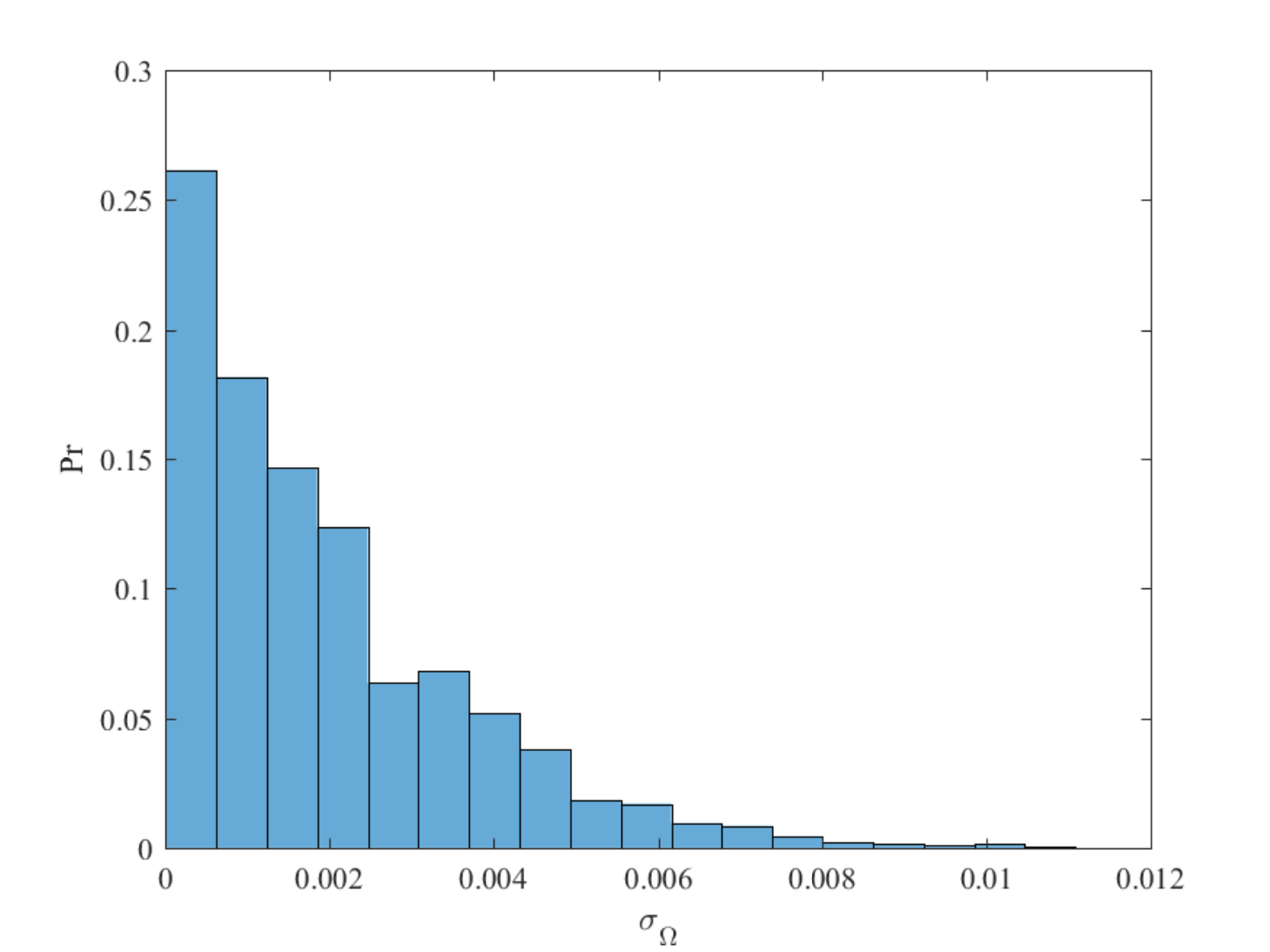}
\caption{Histogram of $\sigma_{\Omega}$, the minimum $\sigma_{N}$ among all sub-matrices of a sensing matrix $\bm \Phi$. The histogram is obtained using 2000 Monte-Carlo trials.}
\label{fig:hist_sigma_omega}
\end{figure}
\subsection{MIP}
Next, we consider the MIP of the sensing matrix. The parameters are set to $N = 64$ and $M = 16$. The frequency codes are uniformly distributed over the discrete set $\mathcal{D}_d$. The relative bandwidths are set to $B/f_c = \{0, 0.1, 0.5 \}$, where $B/f_c = 0$ means that the assumption $\zeta_n \approx 1$ holds. We also simulate the continuous case with $d_n \sim \mathcal{D}_c$ and $\zeta_n \approx 1$. Curves are obtained with $10^6$ Monte-Carlo trials. For each trial, we calculate the mutual coherence $\mu$ of the sensing matrix and depict the corresponding cumulative distribution function. The theoretical bound in (\ref{equ:pr_union}) is also displayed. The results are shown in Fig. \ref{fig:mip}. It can be seen that the theoretical upper bound (\ref{equ:pr_union}) is tight under the assumption that the relative bandwidth is negligible (thus  $\zeta_n \approx 1$ holds). When the relative bandwidth is large, the actual mutual coherence could exceed the predicted one, e.g., in the case that $B/f_c = 0.1$. However, the curve of $B/f_c = 0.5$ is under that of $B/f_c = 0.1$, which indicates that a larger relative bandwidth does not necessarily result in worse mutual incoherence.  
\begin{figure}[!ht]
\centering
\includegraphics[width=2.5in]{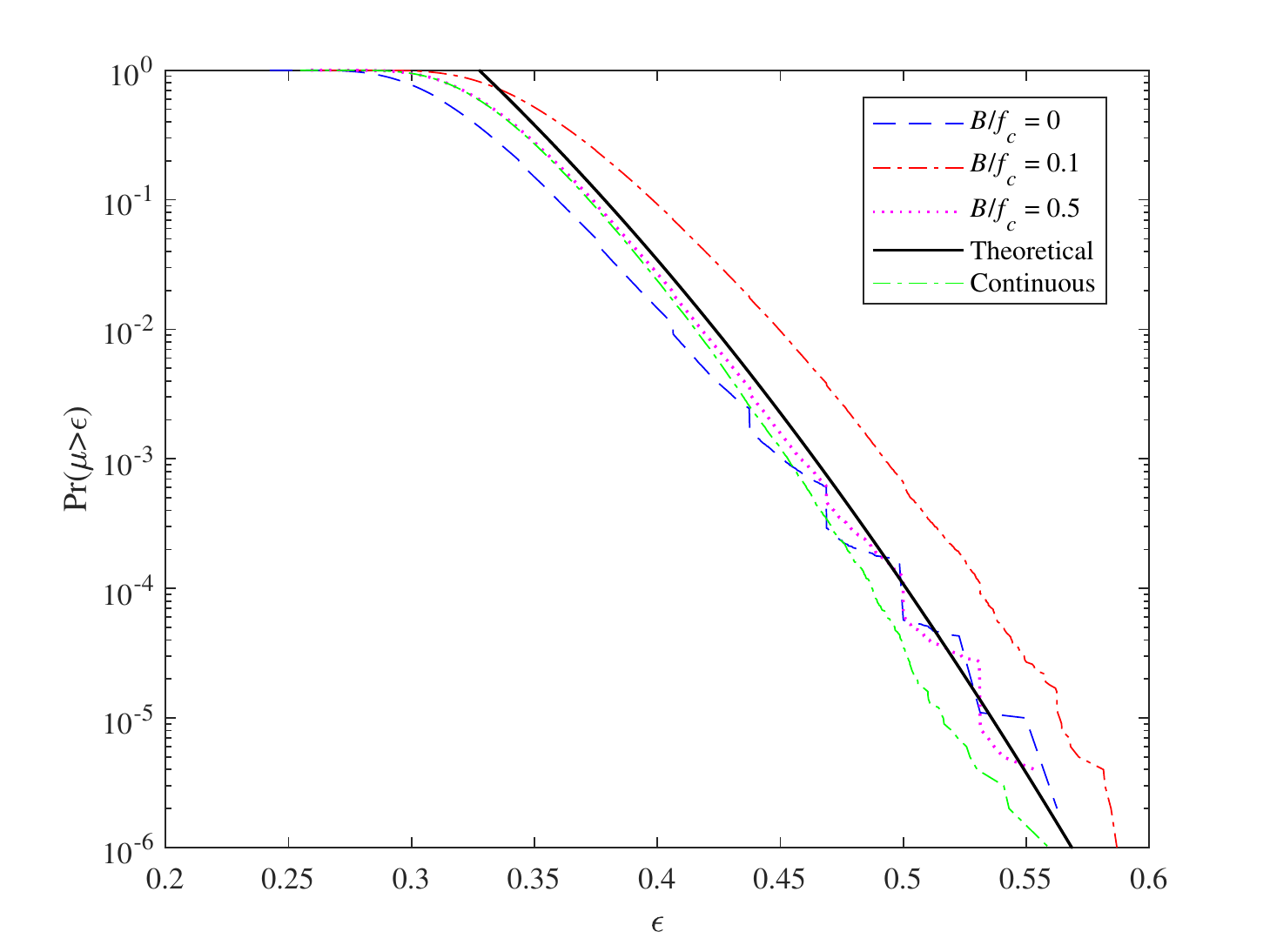}
\caption{The cumulative distribution functions of the mutual coherence $\mu$ obtained from $10^6$ Monte-Carlo trials.}
\label{fig:mip}
\end{figure}
\subsection{Recovery Performance in Noiseless Cases}
In Fig. \ref{fig:phasetransition}, we consider the recovery performance in noiseless cases. In particular, we plot the probability of exact recovery using the basis pursuit algorithm (\ref{equ:P1}) and matched filter (\ref{equ:mf}), where exact recovery means the support set of the unknown vector $\bm x$ is exactly estimated. In the simulations, the pulse number is $N = 64$, the number of frequencies is $M = 8$, and the amplitudes of scattering coefficients are all set to 1. The number of scatterers, $K$, is varied. The initial carrier frequency is $f_c = 10$ GHz and the bandwidth is $B = 64$ MHz. For each point on the curve, we perform 200 Monte-Carlo trials, where the frequency codes are randomly drawn obeying $U(\mathcal{D}_d)$, the support set of $\bm x$ is random, and the phases of non-zeros in $\bm x$ are i.i.d $U([0,2\pi])$. We solve (\ref{equ:P1}) using CVX \cite{cvx, gb08}. In both methods, we assume the number of scatterers, $K$, is known, and the support set is obtained as the indices of the largest $K$ magnitudes in $\bm x$.  The magnitudes are also compared with a threshold $\epsilon = 10^{-2}$. Those not exceeding the threshold are {removed from} the support set. From Fig. \ref{fig:phasetransition}, it is seen that CS dramatically outperforms the traditional matched filter. When $K>5$, the support set recovery probabilities using matched filter drops significantly, because some of the scatterers are masked by the sidelobes. When basis pursuit is applied, the region leading to exact support set recovery in FAR is fairly broad. However, the theoretical bound (\ref{equ:K_rangeMIP}) is 1.5 with $\delta = 0.1$, and is quite pessimistic.
\begin{figure}[!ht]
\centering
\includegraphics[width=2.5in]{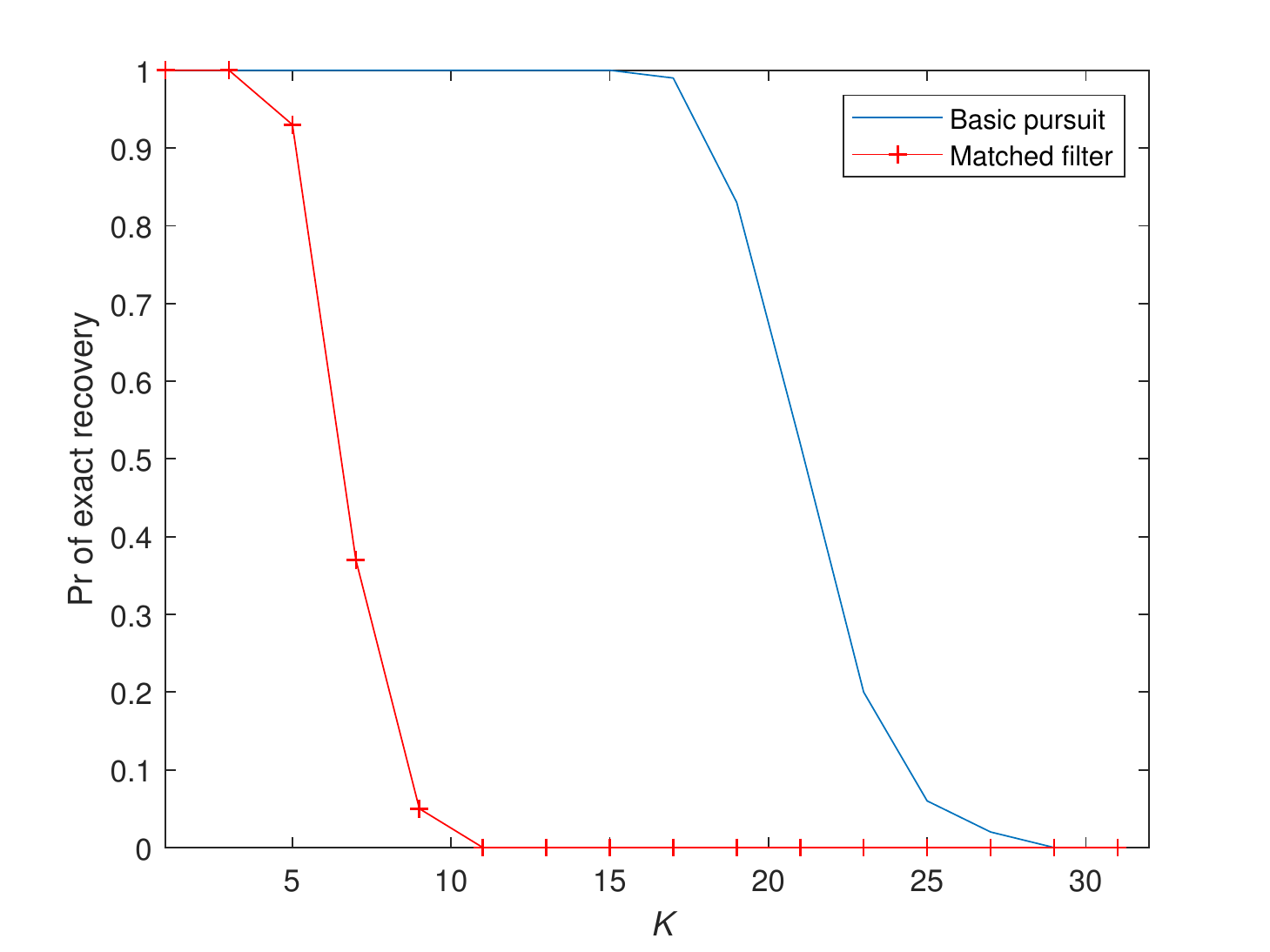}
\caption{Exact support set recovery probabilities using basis pursuit (\ref{equ:P1}) and matched filter (\ref{equ:mf}) in noiseless cases. }
\label{fig:phasetransition}
\end{figure}
\subsection{Recovery Performance in Noisy Cases}
We next consider noisy cases, and choose the probability of successful recovery to evaluate the performance of different CS algorithms. A successful recovery is defined as exact recovery of the support set.  In the simulations, $f_c = 10$ GHz, $B = 64$ MHz, $N = 64$, $M = 8$ and the number of scatterers $K =3$. The scatterers have identical amplitudes of 1 with random phases. 
The noise $\bm w$ in (\ref{equ:noisySigModel}) is assumed Gaussian white noise with a covariance matrix $\sigma^2 {\bm I}$, and $\sigma^2$ varies from -15 dB to 15 dB.  Subspace pursuit \cite{Dai2009} and Lasso are compared. In subspace pursuit, the number of scatterers $K$ is assumed known a priori. The Lasso algorithm solves
\begin{equation}
\min_{\bm x} {\frac{1}{2}\| {\bm y} - {\bm \Phi}{\bm x}\|_2^2} + \lambda \| \bm x\|_1
\end{equation}
with $\lambda = 3\sigma^2$, and is implemented with CVX. When the magnitude of the estimate is larger than $\epsilon = 0.2$, the corresponding index is put into the estimated support set. The results are shown in Fig. \ref{fig:noisy} with 200 Monte-Carlo trials. Both algorithms have high successful recovery probabilities for an FAR in high SNR, $\sigma^2 \leq 0$ dB. Subspace pursuit outperforms Lasso with the genie-aided information on the cardinal number of the support set. We also note that the selection of the parameters $\lambda$ and $\epsilon$ has a significant impact on the performance of Lasso. 
%(full band, small target, matched filter, hit rate)
\begin{figure}[!ht]
\centering
\includegraphics[width=2.5in]{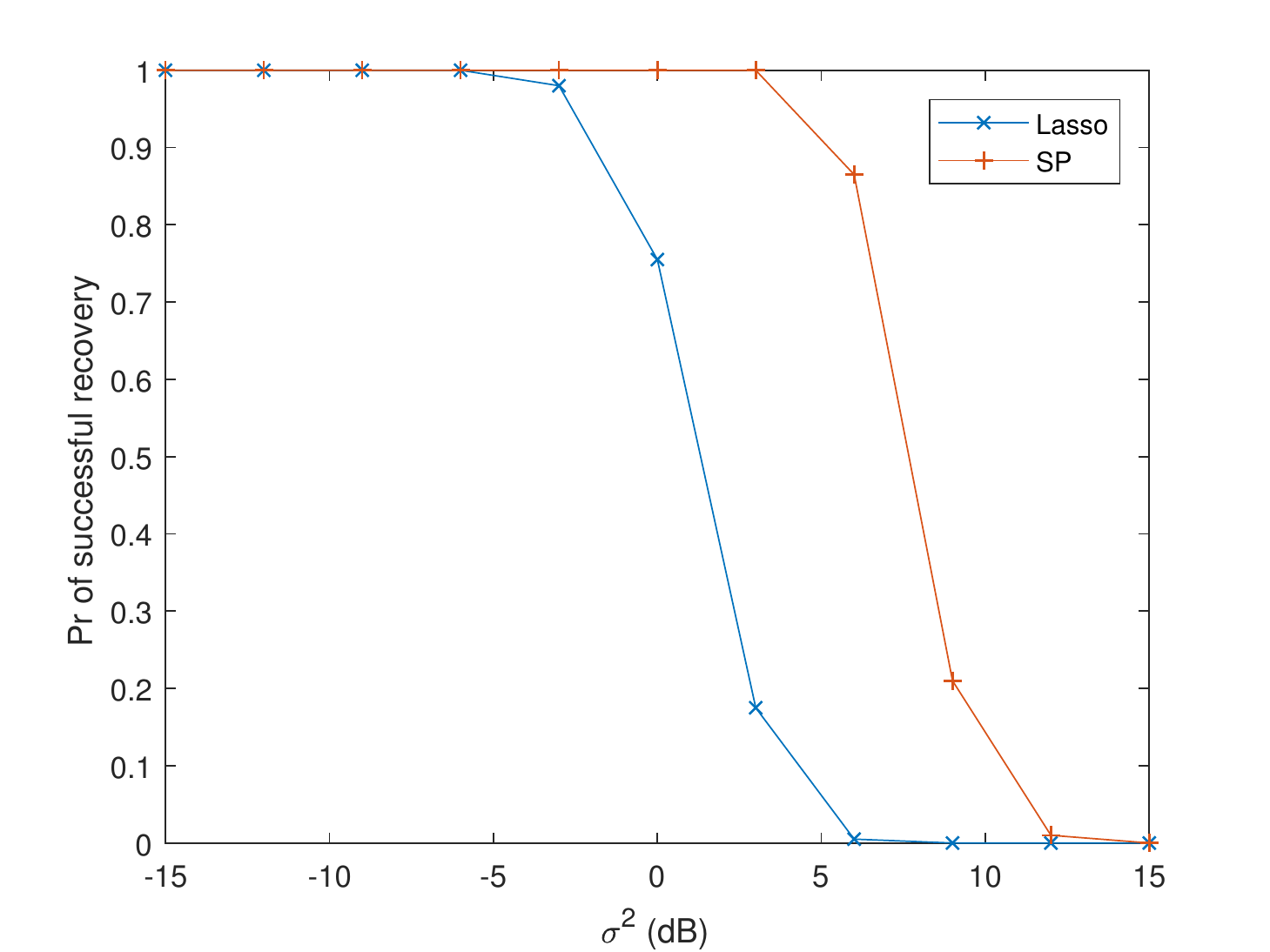}
\caption{Probabilities of successful recovery using subspace pursuit and Lasso algorithms. }
\label{fig:noisy}
\end{figure}
\subsection{Field Experiments}
Next, we show field experiments from a true FAR prototype.  We use separated antennas for transmitting and receiving, respectively, so that returns from short range objects {are} not eclipsed by the transmission.  The radar works at an initial frequency of $f_c=9$ GHz. Frequencies are varied pulse by pulse. There are $M^{\star} = 64$ frequencies with a minimum gap of $16$ MHz, which results in a synthetic bandwidth of $B = 1024$ MHz. {In each pulse, the carrier frequency $f_n = f_c + d_n B $ is randomly chosen. Specifically, $d_n \in \mathcal{D}_d = \left\{0, 1/M^{\star},\dots,  (M^{\star}-1)/M^{\star} \right\}$ and $d_n \sim U(\mathcal{D}_d )$.} The HRR is $c/2B \approx 0.15$ m. We set $T_r = 0.2$ ms and the equivalent pulse duration  {$T_p = 31.25$} ns. The CRR is $cT_p/2 \approx 4.7$ m, and the {number of HRR bins in a CRR bin is $M=\lceil T_p B \rceil =32 < M^{\star}$, which satisfies the condition to eliminate ghost images \cite{Liu2009, Liu2014a}. The} number of pulses is $N = 512$. {The moving target does not cross a CRR bin during the CPI, which requires $v < \frac{T_pc}{2N T_r}\approx 366.2$ m/s. For a slowly moving car, the velocity is lower than 10 m/s.}
%\begin{figure}[!h]
%\centering
%\includegraphics[width=2.5in]{demo}
%\caption{A prototype of FAR. }
%\label{fig:radar}
%\end{figure}
\par
Field experiments are executed to evaluate performance. The target is a household car (see Fig. \ref{fig:field}) with two corner reflectors and four small metal spheres upon its roof to enhance the SNR. When the radar is operated, the car moves {in front of the radar} at a nearly constant speed along the road, surrounded by static objects including a big stone, roadside trees and iron barriers. {Returns from all scatterers located in $(0,cT_r/2]$, i.e. $(0,30]$ km, are collected. There are $\lfloor T_r/T_p \rfloor = 6400$ CRR bins.} We perform static clutter canceling \cite{Axelsson2006}  over all CRR bins. Then in these CRR bins, we find the CRR bin which has the maximum amplitude of radar echoes, and infer that the car is located in that bin. With data in that CRR bin, we perform range-Doppler processing of the target. We apply the matched filter (\ref{equ:mf}) and the OMP algorithm to jointly estimate the HRR profile and Doppler of the target. In OMP, the algorithm iterates 50 times, i.e., assuming $K=50$. The results are shown in Fig. \ref{fig:mf} and Fig. \ref{fig:omp}, respectively. {To better demonstrate the sidelobe, we project the three dimensional images in  Fig. \ref{fig:mf} and Fig. \ref{fig:omp} onto amplitude-range dimensions; see Fig. \ref{fig:hrrp}. In Fig. \ref{fig:hrrp}, only  the maximum $K$ amplitudes in $\hat{\bm x}$ using matched filter are found and shown. } In both methods, the velocity estimation is 6.9 m/s, and the span of the car is around 1.9 m. Comparing the recovery performance, we see that the matched filter suffers from sidelobe pedestal while OMP demonstrates a clearer reconstruction.
\begin{figure}[!ht]
\centering
\includegraphics[width=2.5in]{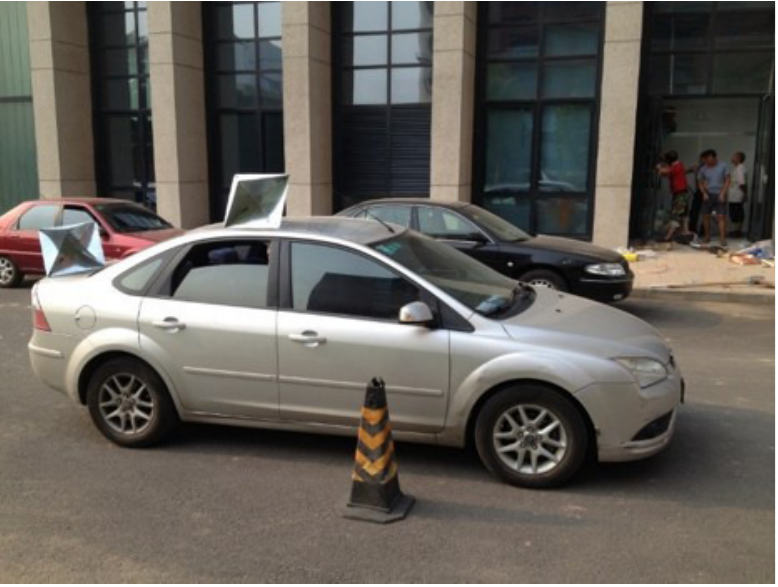}
\caption{Field experiment scenario. }
\label{fig:field}
\end{figure}
\begin{figure}[!ht]
\centering
\includegraphics[width=2.5in]{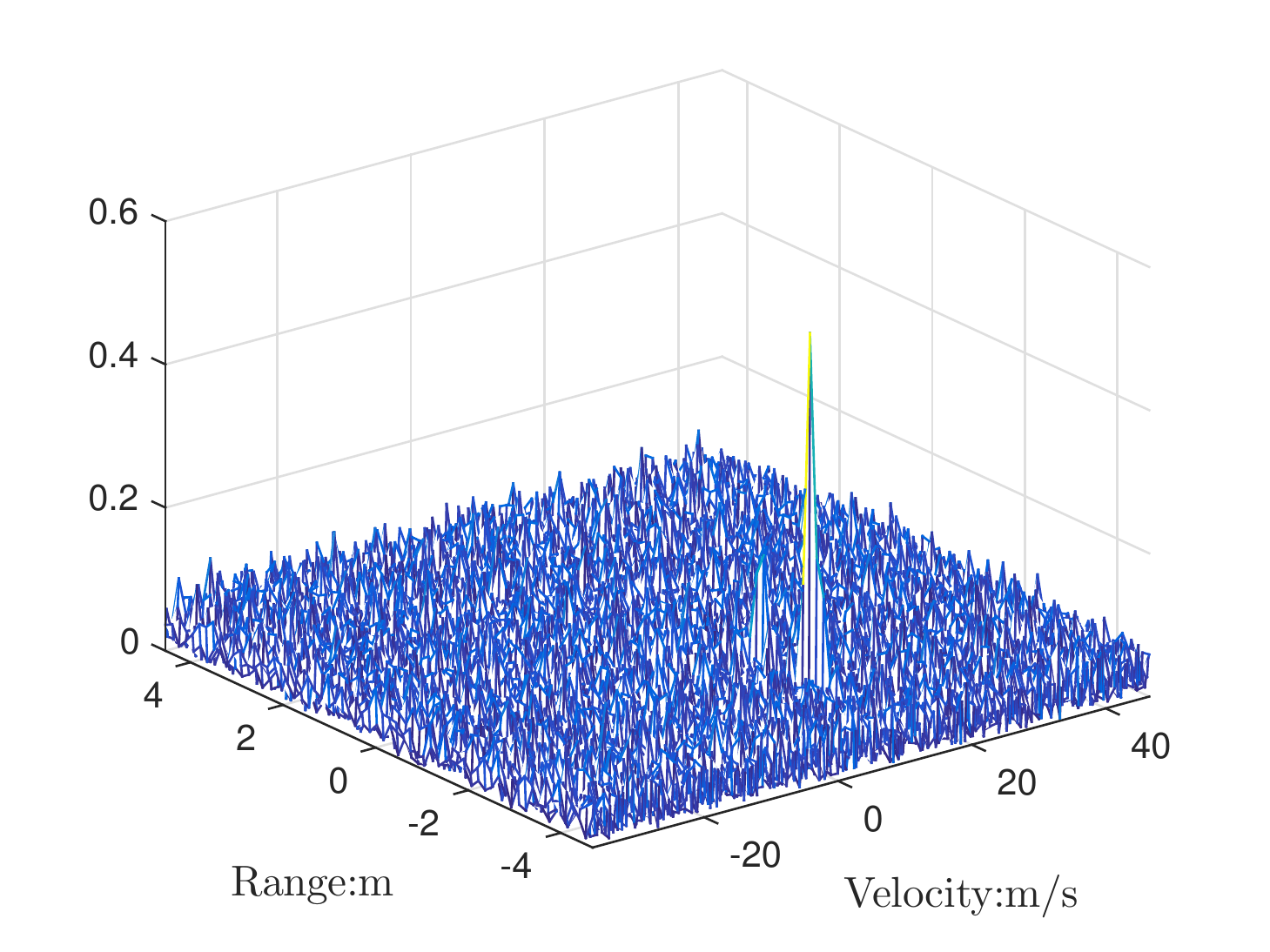}
\caption{Field experiment result using a matched filter.}
\label{fig:mf}
\end{figure}
\begin{figure}[!ht]
\centering
\includegraphics[width=2.5in]{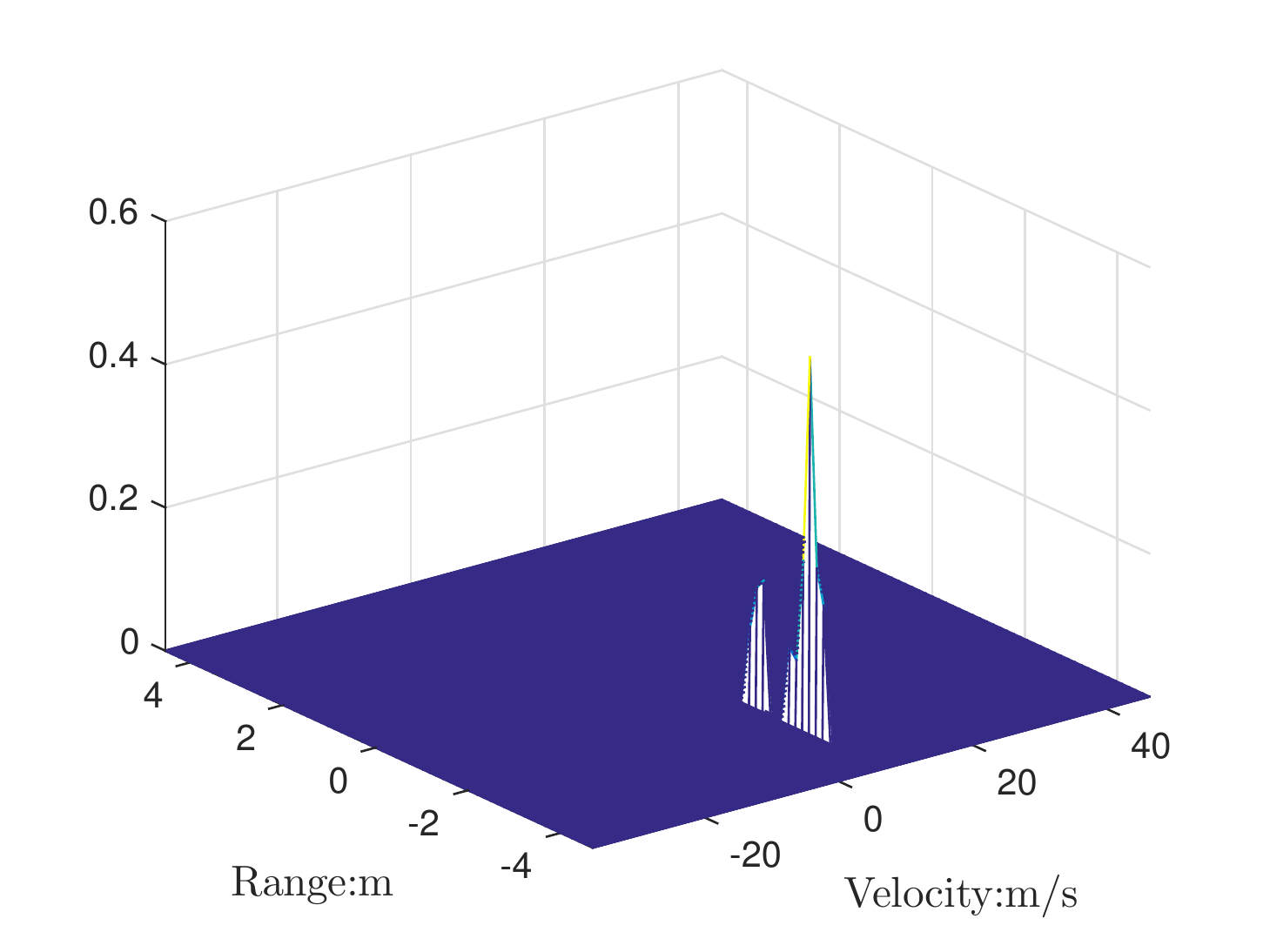}
\caption{Field experiment result using OMP. } %{\color{red} what is the true result: speed, shape.}}
\label{fig:omp}
\end{figure}

\begin{figure}[!ht]
\centering
\includegraphics[width=2.5in]{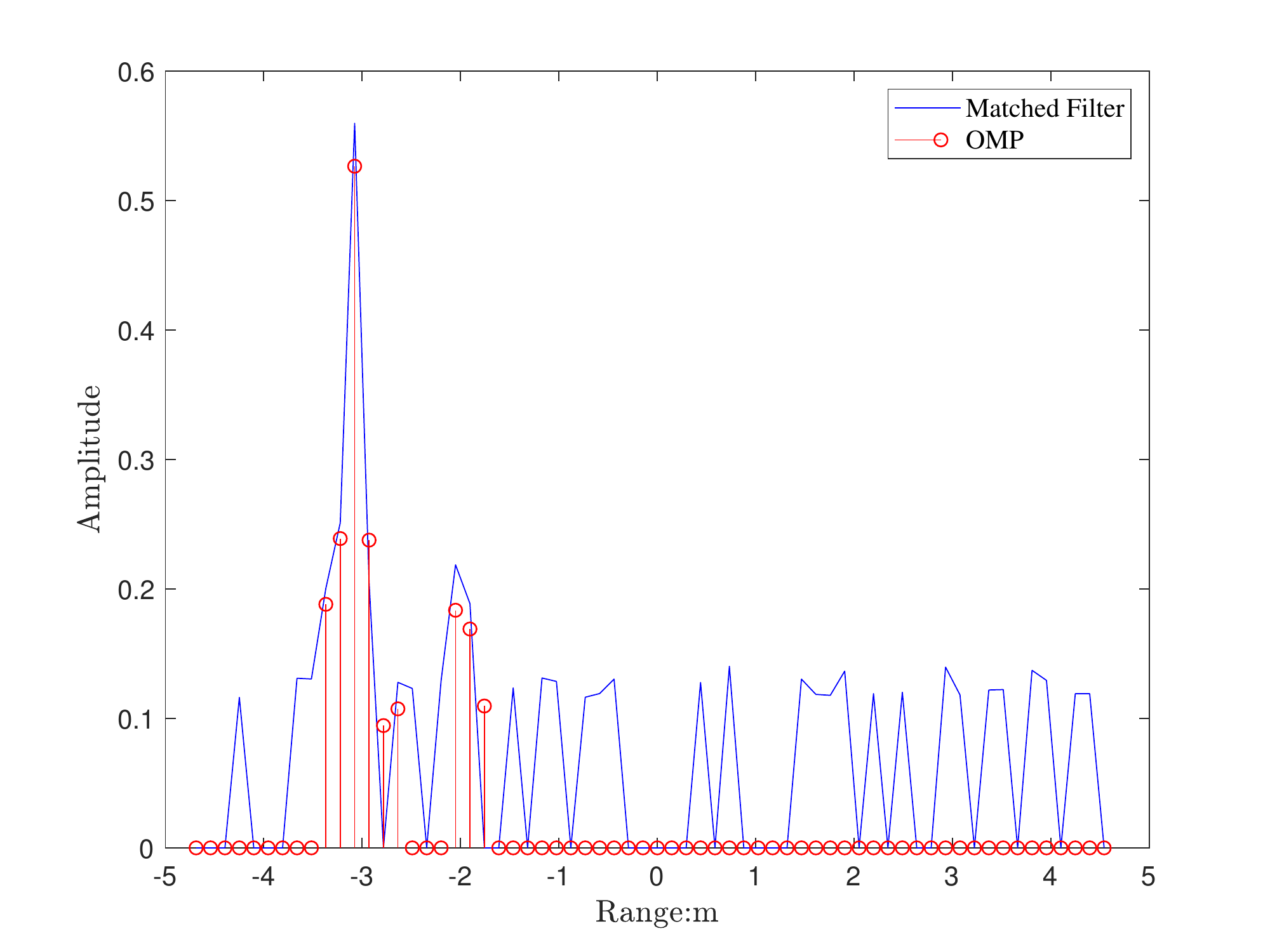}
\caption{Field experiment results projected onto amplitude-range dimensions. } %{\color{red} what is the true result: speed, shape.}}
\label{fig:hrrp}
\end{figure}

%\begin{figure}[!t]
%\centering
%\subfloat[Matched filter]{\includegraphics[width=2.5in]{MF_50}
%\label{fig:mf}}
%\hfil
%\subfloat[OMP]{\includegraphics[width=2.5in]{OMP_50}
%\label{fig:omp}}
%\caption{Field experiment results of an FAR.}
%\label{fig:field_experiment}
%\end{figure}
\section{Conclusion}\label{sec:con}
In this paper, sparse recovery for a frequency agile radar with random frequency codes is studied. We analyzed the spark and mutual incoherence properties of the radar sensing matrix, which guarantee reliable reconstruction of the targets. Using $\ell_0$ minimization, FAR exactly recovers $K=\frac{N}{2}$ scatterers in noiseless cases almost surely, where $N$ is the number of pulses. When we apply $\ell_1$ minimization or greedy CS methods and there is noise, the number of scatterers that is guaranteed to be reliably reconstructed by FAR is on the order of $K = O\left( \sqrt{\frac{N}{\log MN}}\right)$, where $M$ is the number of HRR bins in a CRR bin. Numerical simulations and field experiments were executed to validate the theoretical results and also demonstrate the practical recovery performance of FAR. 

%({\color{red} say explicitly what the findings were  $K$ on the order of ... say also in intro. and say how this related to other radar systems.}) %Explicit conditions for radar parameters and target scenes have been derived such that the radar sensing matrix have these properties with high probabilities. 

% or


\appendices
\section{Proof of Theorem \ref{thm:sparkFAR}}\label{app:spark}
To avoid confusion, in this appendix we will index the row of a matrix by $\xi$ and the column by $\eta$. 
We need to prove that any $N$ columns of $\bm \Phi$ are almost surely linearly independent. 
Following the form in (\ref{equ:phi}), we first fix some $N$ columns $\bm \Phi_{l_\eta+m_\eta N}$, with $\eta$ running through $\{0,1,\ldots, N-1\}$. 
These $N$ columns constitute a new $N\times N$ matrix $\bm A$ whose elements are $\left[{\bm A}\right]_{\xi,\eta}=\left[{\bm \Phi}\right]_{\xi,l_\eta+m_\eta N}$. We need to show that $\bm A$ is almost surely invertible. For brevity, we set
\begin{equation}
  z_\xi:=\exp(j2\pi d_\xi),\quad c_{\xi,\eta}:=\exp(j2\pi l_\eta\xi/N).
\end{equation}
\par
With the notation above, we may write $\left[{\bm A}\right]_{\xi,\eta}=c_{\xi,\eta}z_{\xi}^{m_\eta}$. 
The proof follows from the following three lemmas. 
The first one, being more abstract, is a strengthened version of the well-known fact that the set of zeros of a nonzero polynomial has measure zero.
\begin{lemma}\label{lem:polyrootMeas}
  Let $P\in\mathbb C[z_1,\cdots,z_n]$ be a nonzero complex polynomial in $n$-variables. Denote by $\mathcal N\subset\mathbb C^n$ the set of zeros of $P$. Consider the $n$-torus $\mathbb T^n=\underbrace{S^1\times\cdots\times S^1}_n$ with its obvious embedding $\mathbb T^n\subset\mathbb C^n$. Let $\sigma_n$ be the Haar measure on $\mathbb T^n$. We have
  \begin{equation}
    \sigma_n(\mathcal N\cap\mathbb T^n)=0.
  \end{equation}
\end{lemma}
\begin{proof}
  The left hand side of the above equality is well-defined, 
  since $\mathcal N$ is closed and, consequently, $\mathcal N\cap\mathbb T^n$ is a closed set in $\mathbb T^n$. 
  Note that $\sigma_n$ coincides with the product measure $\underbrace{\sigma_1\times\cdots\times\sigma_1}_n$. This enables us to prove the result by induction via Fubini's theorem. 
  \par
  For $n=1$ the set $\mathcal N$ is discrete, thus the proposition is trivial. 
  Suppose for $n=k$ the proposition is true. 
  Let $\pi:\mathbb C^{k+1}\to\mathbb C$ be the projection onto the first component. 
  We now make the natural identification $\mathbb C[z_1,\cdots,z_{k+1}]\approx(\mathbb C[z_2,\cdots,z_{k+1}])[z_1]$. 
  With the result for $n=k$ in hand, 
  we can find a $\sigma_k$-negligible set {(a set of measure zero)} $\mathcal O\subset\mathbb T^k$ such that 
  for any $(z_2,\cdots,z_{k+1})\in\mathbb T^k\setminus\mathcal O$, 
  the polynomial $P$ (viewed as a polynomial in $z_1$) has at least one nonzero coefficient, 
  i.e. is a nonzero polynomial in $z_1$. 
  Then by the result for $n=1$, for $\mathbf z^{(k)}\in\mathbb T^k\setminus\mathcal O$, the measure of the slice 
  \begin{equation}
    \sigma_1(\pi(\mathcal N\cap\mathbb T^{k+1}\cap\{(z_2,\cdots,z_{k+1})=\mathbf z^{(k)}\}))=0.
  \end{equation}
  We write this fact as $\sigma_1(\pi(\mathcal N|_{\mathbf z^{(k)}}))=0$. 
  Bearing in mind that $\sigma_{k+1}=\sigma_k\times\sigma_1$, we apply Fubini's theorem to obtain
  \begin{equation}
    \begin{split}
        \sigma_{k+1}&(\mathcal N\cap\mathbb T^{k+1})=\int_{\mathbb T^{k+1}}1_{\mathcal N}d\sigma_{k+1}\\
        &=\int_{\mathbb T^k\setminus\mathcal O}d\sigma_k\int_{S^1}1_\mathcal N\cdot\mu(dz_1)\\
        &+
        \int_\mathcal Od\sigma_k\int_{S^1}1_\mathcal N\cdot\sigma_1(dz_1)\\
        &=\int_{\mathbb T^k\setminus\mathcal O}\sigma_1(\pi(\mathcal N|_{\mathbf z^{(k)}}))\sigma_k(d\mathbf z^{(k)})
        +0\\
        &=0,
    \end{split}
  \end{equation}
which completes the proof.
\end{proof}
\begin{lemma}\label{lem:genericSpark}
  Let $\{d_\xi\}$ be independent random variables with continuous distributions for $\xi=0,\ldots,N-1$. 
  Fix some complex numbers $c_{\xi,\eta}$, positive real constants $A_\xi$ and nonnegative integers $m_\eta$, $\xi,\eta=0,\ldots, N-1$. Let $z_\xi=\exp(jA_\xi d_\xi)$. 
  Then the following two statements are equivalent:
  \begin{enumerate}[(i)]
    \item The $N\times N$ random matrix $\bm A$ with elements
    \begin{equation}
      \left[{\bm A}\right]_{\xi,\eta}=c_{\xi,\eta}z_\xi^{m_\eta}
    \end{equation}
    is almost surely invertible.
    \item There exists a vector $\mathbf w\in\mathbb C^N$ such that the $N\times N$ deterministic matrix ${\bm A}^{(\mathbf w)}$ with elements
    \begin{equation}\label{equ:auxInvMatrix}
      \left[{\bm A}^{(\mathbf w)}\right]_{\xi,\eta}=c_{\xi,\eta}w_\xi^{m_\eta}
    \end{equation}
    is invertible.
  \end{enumerate} 
\end{lemma}
\begin{proof}
  (i)$\implies$(ii) is obvious. We prove that (ii)$\implies$(i). 
  Note that $\det\bm A$ is a polynomial $P$ in $N$ variables $z_0,\ldots,z_{N-1}$. 
  By (ii), this polynomial is nonzero, since $P(w_1,\ldots,w_{N-1})=\det{\bm A}^{(\mathbf w)}\ne0$. 
  Let $\mathcal N$ be the set of zeros of $P$. 
  Lemma \ref{lem:polyrootMeas} now implies that $\sigma_N(\mathcal N\cap\mathbb T^N)=0$. 
  On the other hand, the map $\phi:\mathbb R^N\to\mathbb T^N$ defined by 
  $\phi(x_0,\ldots,x_{N-1})=(e^{j2\pi x_0},\ldots,e^{j2\pi x_{N-1}})$ is obviously absolutely continuous. 
  By the assumption that $d_\xi$ are independent and absolutely continuous, 
  the map $(d_0,\ldots,d_{N-1}):\Omega\to\mathbb R^N$ is also absolutely continuous. 
  Thus the probability of the event $(d_1,\ldots, d_{N-1})^{-1}\phi^{-1}(\mathcal N\cap\mathbb T^N)$ is $0$, 
  as desired. 
\end{proof}
\begin{lemma}\label{lem:nondegenPoly}
  Let $l_\eta,m_\eta$ be nonnegative integers and $c_{\xi,\eta}=\exp(j2\pi l_\eta\xi/N)$. 
  Furthermore, assume that the map $\eta\mapsto(\exp(j2\pi l_\eta/N), m_\eta)$ is injective 
  when $\eta$ takes value in $\{0,\ldots,N-1\}$. 
  Then the statement (ii) in Lemma \ref{lem:genericSpark} is true.
\end{lemma}
\begin{proof}
  Choose some real number $b$ which is not a rational multiple of $\pi$. 
  Take $w_\xi=\exp(jb\xi)$. Then ${\bm A}^{(\mathbf w)}$ becomes a Vandermonde matrix
  \begin{equation}
    \left[{\bm A}^{(\mathbf w)}\right]_{\xi,\eta}=e^{j\xi(2\pi l_\eta/N+bm_\eta)}.
  \end{equation}
  \par
  By the determinant of a Vandermonde matrix, it suffices to prove that 
  $\eta\mapsto\exp(j(2\pi l_\eta/N+bm_\eta))$ is injective. 
  Suppose to the contrary that for some $\eta\ne\eta'$ 
  we have $(2\pi l_\eta/N+bm_\eta)-(2\pi l_{\eta'}/N+bm_{\eta'})=2k\pi$, $k\in\mathbb Z$. 
  Since $b$ is not a rational multiple of $\pi$, this implies $m_\eta=m_{\eta'}$, 
  { henceforth} $2\pi l_\eta/N-2\pi l_{\eta'}/N=2k\pi$. 
  But then $(\exp(2\pi l_\eta/N),m_\eta)=(\exp(2\pi l_{\eta'}/N),m_{\eta'})$, 
  contradicting the injectivity of $\eta\mapsto(\exp(j2\pi l_\eta/N), m_\eta)$. 
  This proves that
  \begin{equation}
    \det{\bm A}^{(\mathbf w)}=\prod_{0\le\eta<\eta'<N}
    \left(e^{j\frac{2\pi }{N}l_{\eta'}+jbm_{\eta'}}-e^{j\frac{2\pi}{N} l_\eta+jbm_\eta}\right)
    \ne 0.
  \end{equation}
  In other words, our choice of $\mathbf w\in\mathbb C^N$ makes the matrix in (\ref{equ:auxInvMatrix}) invertible.
\end{proof}
\par
The map $\eta\mapsto(\exp(j2\pi l_\eta/N), m_\eta)$ is injective since $(l_\eta,m_\eta)$ is pairwise distinct and $0\le l_\eta<N$. Then one may apply Lemma \ref{lem:nondegenPoly}  and Lemma \ref{lem:genericSpark} successively and, by the consequence of Lemma 12, conclude the proof of Theorem \ref{thm:sparkFAR}. %{\color{red}(missing some connecting words here ...)} Lemma \ref{lem:nondegenPoly} and Lemma \ref{lem:genericSpark} imply Theorem \ref{thm:sparkFAR} directly. 
\par
We conclude with a remark on the robustness of our proof. The approximation we take here is $\zeta_\xi\approx 1$. Since $\zeta_\xi$ is absorbed in $c_{\xi,\eta}$, only the conclusion of Lemma \ref{lem:nondegenPoly} will be affected if $\zeta_\xi\ne1$. However, from the proof of Lemma \ref{lem:nondegenPoly}, for $\zeta_\xi$ sufficiently close to $1$, its conclusion remains true, thus Theorem \ref{thm:sparkFAR} still holds.
\section{Proof of Lemma \ref{lem:asymGaussian}}
\label{app:asymGaussian}
We divide the proof into two parts: in the first part, we use Lyapunov's condition to prove that the real and imaginary parts of $\chi_l$ have a joint Gaussian distribution asymptotically; in the second part, the expectation and variance are calculated. For conciseness, we omit the subscript $l$ in $\chi_l$. The parameters $p$ and $q$ belong to specific grids as stated in Subsection \ref{subsec:sm_matrix}, respectively. Note that $l \in \Xi$, which means $p \neq 0$. Also recall the assumption that random frequency codes $d_n \sim U\left( \mathcal{D}_d\right)$, and are independent from each other. 
\subsection{Asymptotic distribution}
{First, consider the case $p \neq \pi$.} Introduce a constant $\lambda$, and define a random variable
\begin{equation}
\begin{split}
X_n :&= {\rm Re}(e^{jpMd_n+jqn}) + \lambda {\rm Im}(e^{jpMd_n+jqn})\\
&=\cos(pMd_n+qn) + \lambda \sin(pMd_n+q{n}) \\
&=T_{\lambda}\cos(pMd_n+qn+z_{\lambda}) \\
&=T_{\lambda}\cos(pMd_n+\theta),
\end{split}
\end{equation}
where $T_{\lambda} := \sqrt{1+\lambda^2}$, $\sin z_{\lambda} = - \frac{\lambda}{T_{\lambda}}$, $\cos z_{\lambda} =  \frac{1}{T_{\lambda}}$ and $\theta:=qn+z_{\lambda}$. Define
\begin{equation}
Y_N := \sum \limits_{n=0}^{N-1} X_n.
\end{equation}
Note that $Y_N = \chi$ when $\lambda = j$. Thus, it is equivalent to prove that for any real value $\lambda$, as $N \rightarrow \infty$, it holds that
\begin{equation}
\label{equ:gaussian01}
\frac{Y_N-{\rm E}[{Y_N}]}{S_N} \sim \mathcal{N}(0,1),
\end{equation}
where $S_N$ is the standard variance, obeying
\begin{equation}
\begin{split}
S_N^2 &= {\rm E}\left[\left( \sum \limits_{n=0}^{N-1}\left( X_n -{\rm E}[X_n] \right) \right)^2\right]\\
%&=\sum \limits_{n=0}^{N-1}{\rm E}\left[ ( X_n -{\rm E}[X_n] )^2\right] + 2\sum \limits_{0\leq m <n \leq N-1}{\rm E}\left[\left( X_m -{\rm E}[X_m] \right)\left( X_n -{\rm E}[X_n] \right)\right]\\
&=\sum \limits_{n=0}^{N-1}{\rm E}\left[( X_n -{\rm E}[X_n])^2 \right],
\end{split}
\end{equation}
where independence between {the} $X_n$ is used. For $p \in (0,2\pi)$ and $M>1$, it holds that $S_N^2 > 0$.%{\color{red} Need to better organize the proof. Say what you want to prove and then go in order.}
\par
According to Lyapunov's central limit theorem \cite{Borovkov2013}, (\ref{equ:gaussian01}) holds, if for some $\delta >0$,
\begin{equation}\label{equ:Lyapunov}
\lim \limits_{N \rightarrow \infty}{ \frac{\sum \limits_{n = 0}^{N-1} {\rm E}\left[ \left|X_n-{\rm E}[X_n] \right|^{2+\delta} \right]}{S_N^{2+\delta}}} = 0.
\end{equation}
We consider $\delta = 1$. In the following, we calculate ${\rm E}\left[ \left|X_n-{\rm E}[X_n] \right|^3\right]$ and $S_N^{3}$ to verify that (\ref{equ:Lyapunov}) holds {assuming} $d_n \sim U\left( \mathcal{D}_d \right)$. 
\par
To calculate $S_N^{3}$,  derive 
\begin{equation}
\begin{split}
{\rm E} &[X_n] = \frac{T_{\lambda}}{M} \sum_{m = 0}^{M-1} \cos (pm+\theta)\\
&=\frac{T_{\lambda}}{2M\sin \frac{p}{2}}\left(
\sin \left(\frac{M-1}{2}p + \theta  \right) - \sin \left(-\frac{p}{2} + \theta \right) \right) \\
&= \frac{T_{\lambda}}{M\sin \frac{p}{2}} \sin \frac{Mp}{2}\cos\left(\frac{M-1}{2}p+\theta \right),
\end{split}
\end{equation}
where we assume $p \neq 0$. In addition,
\begin{equation}
\begin{split}
{\rm E} &[X_n^2] = \frac{T_{\lambda}^2}{M} \sum_{m = 0}^{M-1} \cos ^2 (pm+\theta)\\
&=\frac{T_{\lambda}^2}{2M} \sum_{m = 0}^{M-1} \left( \cos \left(
2pm + 2\theta \right) + 1\right) \\
&= \frac{T_{\lambda}^2}{2} + \frac{ T_{\lambda}^2\sin \left( (2M-1)p+2\theta\right)- \sin (2\theta-p)}{4M\sin p} \\
&= \frac{T_{\lambda}^2}{2} + \frac{T_{\lambda}^2\sin(Mp)\cos \left( (M-1)p +2\theta\right)}{2M\sin p}.
\end{split}
\end{equation}
Therefore,
\begin{equation}
\begin{split}
{\rm D} &[X_n] = {\rm E} [X_n^2] - ({\rm E} [X_n])^2 \\
&= \frac{T_{\lambda}^2}{2} + \frac{T_{\lambda}^2\sin (Mp) \cos ((M-1)p+2\theta)}{2M\sin p} \\
&- \frac{T_{\lambda}^2\sin^2\frac{Mp}{2} \cos^2 (\frac{M-1}{2}p+\theta)}{M^2 \sin^2 \frac{p}{2}}\\
&=\frac{T_{\lambda}^2}{2}  + \frac{T_{\lambda}^2\sin (Mp) \cos ((M-1)p+2\theta)}{2M\sin p} \\
&- \frac{T_{\lambda}^2\sin^2\frac{Mp}{2} \cos ((M-1)p+2\theta)}{2M^2 \sin^2 \frac{p}{2}}-\frac{T_{\lambda}^2\sin^2\frac{Mp}{2}}{2M^2 \sin^2 \frac{p}{2}}.
\end{split}
\end{equation}
Applying  $p \in \left\{\frac{2\pi}{M},\frac{2\pi \cdot 2}{M},\dots,\frac{2\pi(M-1)}{M}\right\}$, we have $\sin (Mp) = \sin^2\frac{Mp}{2}= 0$ and $\sin p \neq 0$, $ \sin \frac{p}{2} \neq 0$. Then 
\begin{equation}
S_N^2 = \sum \limits_{n=0}^{N-1} {\rm D}[X_n] = \frac{NT_{\lambda}^2}{2}.
\end{equation}
We {conclude} that $S_N^2 = O(N)$, and $S_N^3 = O(N^{\frac{3}{2}})$. 
\par
To calculate the  numerator  in (\ref{equ:Lyapunov}), note that
\begin{equation}
\left| X_n \right|<C_1,
\end{equation}
\begin{equation}
\left| {{\rm E}} \left[ X_n \right] \right|<C_2,
\end{equation}
where $C_1$ and $C_2$ are positive constants not related to $N$. Then,
\begin{equation}\label{equ:E3}
{\rm E}\left[ \left| X_n - {\rm E}[X_n]\right|^3\right] \leq {\rm E}\left[ \left( \left| X_n \right| + \left|{\rm E}[X_n] \right|\right)^3\right]
 \leq (C_1 + C_2)^3.
\end{equation}
Combing $S_N^3 = O(N^{\frac{3}{2}})$ and (\ref{equ:E3}), we have
\begin{equation}
\lim \limits_{N \rightarrow \infty} \frac{\sum \limits_{n = 0}^{N-1} {\rm E}\left[\left| X_n-{\rm E}[X_n]\right|^3 \right]}{S_N^3}
\leq  \lim \limits_{N \rightarrow \infty} \frac{N(C_1 + C_2)^3}{O(N^\frac{3}{2})} = 0.
\end{equation}
Thus, (\ref{equ:Lyapunov}) holds.
\par
{When $p = \pi$, $e^{jpMd_n+jqn} = e^{j\pi Md_n+jqn} = (-1)^{Md_n}e^{jqn}$. Define a random variable 
\begin{equation}
X_n^{'}:= T_{\lambda} (-1)^{Md_n}\cos \theta.
\end{equation}
Following similar steps as above, we find that (\ref{equ:Lyapunov}) still hold for $X_n^{'}$.
}
\par
 According to Lyapunov's central limit theorem, as $N \rightarrow \infty$, ${\rm Re}(\chi)$ and ${\rm Im}(\chi)$ have an asymptotic joint Gaussian distribution.
\subsection{Expectation and variance}
In this subsection, we calculate the expectations and variances of ${\rm Re}(\chi)$ and ${\rm Im}(\chi)$. Denote the variances of the real and imaginary parts and  the correlation coefficient as $\sigma_1^2$, $\sigma_2^2$, and $\sigma_{12}$, respectively, i.e.
\begin{equation}
\left[
\begin{array}{c}
{\rm Re}\left( \chi \right) \\
{\rm Im}\left( \chi \right) \\
\end{array}
 \right] \sim \mathcal{N}
\left( \left[
\begin{array}{c}
{\rm Re}\left( {\rm E} [ \chi] \right)\\
{\rm Im}\left({\rm E} [ \chi] \right)\\
\end{array}
\right],\left[
\begin{array}{cc}
\sigma_1^2 & \sigma_{12}\\
\sigma_{12} & \sigma_2^2\\
\end{array}
\right]
\right).
\end{equation}
\par
We start by analyzing the expectation of the complex valued $\chi$,
\begin{equation}
{\rm E}\left[\chi\right] = {\rm E}\left[ \frac{1}{N} \sum \limits_{n = 0}^{N-1} e^{jpMd_n + jqn}\right].
\end{equation}
Since Pr$(d_n = \frac{m}{M}) = \frac{1}{M}$, it holds that
\begin{equation}
{\rm E}\left[\chi\right] = \sum \limits_{m=0}^{M-1} \frac{1}{MN} \sum \limits_{n = 0}^{N-1} e^{jpm + jqn} .
\end{equation}
Exchanging the order of summations,
\begin{equation}
\begin{split}
{\rm E}\left[\chi\right] &= \sum \limits_{n=0}^{N-1} \frac{1}{MN}e^{jqn} \sum \limits_{m = 0}^{M-1} e^{jpm}\\
&= \frac{1}{MN} \frac{1-e^{jpM}}{1-e^{jp}}\sum \limits_{n=0}^{N-1}e^{jqn}\\
&=0,
\end{split}
\end{equation}
where the last equality holds because $p \in \left\{\frac{2\pi}{M},\frac{2\pi \cdot 2}{M},\dots,\frac{2\pi(M-1)}{M}\right\}$, which implies $e^{jpM} = 1$ while $e^{jp} \neq 1$ and hence $  \frac{1-e^{jpM}}{1-e^{jp}} = 0$. 
\par
Next, we calculate the variances $\sigma_1^2$, $\sigma_2^2$ and $\sigma_{12}$. According to \cite{Lo1964}, it holds that
\begin{equation}\label{equ:e_xx}
{\rm E}\left[ \chi^2 \right] = \sigma_1^2 - \sigma_2^2 + 2j \sigma_{12},
\end{equation}
\begin{equation}\label{equ:e_xxH}
{\rm E}\left[ \left|\chi\right|^2 \right] = \sigma_1^2 + \sigma_2^2.
\end{equation}
The left {hand} side of (\ref{equ:e_xx}) satisfies 
\begin{equation}
\begin{split}
{\rm E}\left[ \chi^2 \right] &= {\rm E}\left[ \frac{1}{N^2} \sum \limits_{n = 0}^{N-1}e^{jpMd_n + jqn} \sum \limits_{k = 0}^{N-1}e^{jp{M}d_k + jqk}\right]\\
&= \frac{1}{N^2} \sum \limits_{n =0}^{N-1}  \sum \limits_{k =0, k \neq n}^{N-1} e^{jq(n+k)} {\rm E} \left[e^{jpM(d_n+d_k)} \right]\\
&+ \frac{1}{N^2} \sum \limits_{n = 0}^{N-1} e^{j2qn} {\rm E} \left[e^{j2pMd_n} \right].
\end{split}
\end{equation}
Applying ${\mathbb{P}}\left(d_n = \frac{m}{M}\right) = \frac{1}{M}$ and independence between $d_n$, 
\begin{equation}\label{equ:Echi2}
\begin{split}
{\rm E}\left[ \chi^2 \right] &=\frac{1}{N^2} \sum \limits_{n =0}^{N-1}  \sum \limits_{k =0, k \neq n}^{N-1} e^{jq(n+k)} \sum_{m_1 =0}^{M-1} \sum_{m_2 =0}^{M-1} \frac{1}{M^2}e^{jp(m_1+m_2)}\\
&+ \frac{1}{N^2} \sum \limits_{n = 0}^{N-1} e^{j2qn} \sum \limits_{m=0}^{M-1}\frac{1}{M}e^{jp2m}\\
&=\frac{1}{N^2M^2}  \frac{(1-e^{jpM})^2}{(1-e^{jp})^2}\sum \limits_{n =0}^{N-1}  \sum \limits_{k =0, k \neq n}^{N-1} e^{jq(n+k)}\\
&+ \frac{1}{N^2M} \frac{1-e^{j2pM}}{1-e^{j2p}}\sum \limits_{n = 0}^{N-1} e^{j2qn} \\
&=\frac{1}{N^2M^2}\frac{(1-e^{jpM})^2}{(1-e^{jp})^2}\left( \frac{(1-e^{jNq})^2}{(1-e^{jq})^2} - \frac{1-e^{j2qN}}{1-e^{j2q}} \right) \\
&+ \frac{1}{N^2M} \frac{1-e^{j2pM}}{1-e^{j2p}}\frac{1-e^{j2qN}}{1-e^{j2q}} .
\end{split}
\end{equation}
According to the assumption  $p \in \left\{\frac{2\pi}{M},\frac{2\pi \cdot 2}{M},\dots,\frac{2\pi(M-1)}{M}\right\}$, we have $\frac{1-e^{jpM}}{1-e^{jp}} = 0$ and thus the first term in (\ref{equ:Echi2}) equals zero. Note that
\begin{equation}
\lim_{x\rightarrow \pi} \frac{1-e^{j2Mx}}{1-e^{j2x}} = M.
\end{equation}
We conclude that
\begin{equation}
{{\rm E}}\left[ \chi ^2 \right] =\begin{cases}
	\frac{1}{N},\ \text{if\ }p=q=\pi ,\\
	0,\ \text{otherwise}.\\
\end{cases}
\end{equation}
Similarly, as for the left side of (\ref{equ:e_xxH}), we have
\begin{equation}
\begin{split}
&{\rm E}\left[ |\chi|^2 \right] = {\rm E}\left[ \frac{1}{N^2} \sum \limits_{n = 0}^{N-1}e^{jpMd_n + jqn} \sum \limits_{k = 0}^{N-1}e^{-jpMd_k - jqk}\right]\\
&= \frac{1}{N^2}  \sum \limits_{n =0}^{N-1}  \sum \limits_{k =0, k \neq n}^{N-1} e^{jq(n-k)} {\rm E} \left[e^{jpM(d_n-d_k)} \right]
+ \frac{\sum \limits_{n = 0}^{N-1} {\rm E} \left[1 \right]}{N^2} \\
&=\frac{1}{N^2}  \sum \limits_{n =0}^{N-1}  \sum \limits_{k =0, k \neq n}^{N-1} e^{jq(n-k)} \sum_{m_1 =0}^{M-1} \sum_{m_2 =0}^{M-1} \frac{e^{jp(m_1-m_2)}}{M^2}
+ \frac{1}{N}\\
&=\frac{(1-e^{jpM})(1-e^{-jpM})}{N^2M^2(1-e^{jp})(1-e^{-jp})}  \sum \limits_{n =0}^{N-1}  \sum \limits_{k =0, k \neq n}^{N-1} e^{jq(n-k)}
+ \frac{1}{N} \\
&=\frac{1}{N^2M^2}\frac{\left|1-e^{jpM}\right|^2}{\left|1-e^{jp}\right|^2}\left( \frac{\left|1-e^{jNq}\right|^2}{\left|1-e^{jq}\right|^2} - N \right) 
+ \frac{1}{N}\\
&=\frac{1}{N}.
\end{split}
\end{equation}
Substituting ${{\rm E}}\left[ \chi^2 \right] =0$ and ${{\rm E}}\left[ |\chi|^2 \right] =\frac{1}{N}$ into (\ref{equ:e_xx}) and (\ref{equ:e_xxH}), respectively, one finds that $\sigma_1^2 = \sigma_2^2 = \frac{1}{2N}$ and $\sigma_{12}=0$. As for the case $p=q=\pi$, ${{\rm E}}\left[ \chi^2 \right] = \frac{1}{N}$ and ${{\rm E}}\left[ |\chi|^2 \right] =\frac{1}{N}$, it holds that $\sigma_1^2 ={\frac{1}{N}}$ and $\sigma_2^2 = \sigma_{12}=0$.

\section*{Acknowledgment}
The authors would like to thank Mr. Pan Li, Dr. Hailong Shi and Mr. Tong Zhao for providing insightful suggestions, and Mr. Lei Wang for collecting data in the field experiments and performing some simulations.

\ifCLASSOPTIONcaptionsoff
  \newpage
\fi



%\bibliographystyle{IEEEtran}
\bibliographystyle{IEEEtran}
\bibliography{IEEEabrv,IEEEexample}

\end{document}